\documentclass[]{llncs}

\usepackage{xcolor}

\usepackage{amsmath}
\usepackage{amssymb}
\usepackage{url}
\usepackage{alltt}
\usepackage{wrapfig}

\usepackage{tikz}
\usetikzlibrary{automata,arrows,backgrounds,shapes}
\usetikzlibrary{decorations.pathmorphing}
\usetikzlibrary{shapes.arrows,chains}

\usepackage{cmap}
\usepackage{mathtext}

\usepackage[utf8]{inputenc}
\usepackage[english]{babel}

\usepackage[unicode,pdfborder={0 0 0}]{hyperref}

\usepackage{shuffle}
\usepackage{syntax}

\usepackage{caption}
\usepackage{subcaption}

\renewcommand{\subparagraph}{}
\usepackage[small,compact]{titlesec}
\usepackage{enumitem}

\usepackage{sty/comments}
\usepackage{sty/mathematics}
\usepackage{sty/programs}
\usepackage{sty/pnets}

\newcommand{\acks}{\section*{Acknowledgements}}

\renewenvironment{quote}{%
  \list{}{%
    \leftmargin0.3cm   
    \rightmargin\leftmargin
  }
  \item\relax
}
{\endlist}

\begin{document}

\title{Checking Robustness against TSO}

\author{Ahmed Bouajjani\inst{1} \and Egor Derevenetc\inst{2,3} \and Roland Meyer\inst{3}}
\institute{$^1$LIAFA, University Paris 7 \quad $^2$Fraunhofer ITWM \quad $^3$University of Kaiserslautern}

\maketitle
\begin{abstract}
We present algorithms for checking and enforcing robustness of concurrent
programs against the Total Store Ordering (TSO) memory model.
A program is robust if all its TSO computations correspond to computations under the Sequential
Consistency (SC) semantics.

We provide a complete characterization of non-robustness in terms of so-called
attacks: a restricted form of (harmful) out-of-program-order executions.
Then, we show that detecting attacks can be parallelized, and can be solved using state reachability queries under SC semantics in a suitably instrumented program obtained by a linear size source-to-source translation.
Importantly, the construction is valid for an arbitrary number of addresses and an
arbitrary number of parallel threads, and it is independent from the data
domain and from the size of store buffers in the TSO semantics.
In particular, when the data domain is finite and the number of addresses is fixed,
we obtain decidability and complexity results for robustness, even for an arbitrary number of threads.

As a second contribution, we provide an algorithm for computing an optimal set of fences that enforce robustness.
We consider two criteria of optimality: minimization of program size and maximization of its performance.
The algorithms we define are implemented, and we successfully applied them to analyzing and correcting several concurrent algorithms.
\end{abstract}

\everymath{\displaystyle}
\section{Introduction}

Sequential consistency (SC) \cite{Lamport79} is a natural shared-memory model where the actions of different threads are interleaved while the program order between actions of each thread is preserved. However, for performance reasons, modern multiprocessors implement weaker memory models relaxing program order. For instance, the common store-to-load relaxation, which allows loads to overtake stores, reflects the use of {\em store buffers}. It is actually the main feature of the TSO (Total Store Ordering) model adopted, e.g., in x86 machines \cite{SewellCACM2010}.

Nonetheless, most programmers usually assume that memory accesses are performed according to SC where all shared-memory accesses are instantaneous and atomic.
This assumption is actually safe for {\em data-race-free} programs \cite{AdveHillDRF1993},
but in many situations data-race-freedom does not apply. This is, for instance, the case of programs implementing synchronization operations, concurrency libraries, and other performance-critical system services employing lock-free synchronization.
In most cases, programmers design programs that are {\em robust} against relaxations, i.e., for which relaxations do not introduce behaviors that are not allowed under SC. Memory fences must be included appropriately in programs in order to prevent non-SC behaviors. Getting such programs right is a notoriously difficult and error-prone task.
Therefore, important issues in this context are (1) checking robustness of a program against relaxations of a memory model, and (2) identifying a set of program locations where it is {\em necessary} to insert fences to ensure robustness.

In this paper we address these two issues in the case of TSO.
We consider a general setting without fixed bounds on the shared memory size, nor on the size of the store buffers in the TSO semantics, nor on the number of threads. This allows us to reason about robustness of general algorithms without assuming any fixed values for these parameters that depend on the actual machine's implementation.
Moreover, we tackle these issues for general programs, independently from the domain of data they manipulate.

Robustness against memory models has been addressed first by Burckhardt and Musuvathi in \cite{burckhardt-musuvathi-CAV08} (actually, for TSO only), and subsequently by Burnim et al. in \cite{Sen2011}.
Alglave and Maranget developed a general framework for reasoning about robustness against memory models in \cite{Alglave2010,AlglaveM11} (where the term {\em stability} is used instead of robustness). Roughly, these works are based on characterizing robustness in terms of acyclicity of a suitable happens-before relation. In that, they follow the spirit of Shasha and Snir \cite{ShashaSnir88} who introduced a notion of {\em trace} that captures the control and data dependencies between events of an SC computation, and established that computations that are not SC have a happens-before relation 
that contains a cycle. We adopt here the same notion of robustness, i.e., a program is (trace-)robust if each of its TSO computations has the same trace as some of its SC computations.

From an algorithmic point of view, the existing works mentioned above {\em do not provide  decision procedures} for robustness. \cite{burckhardt-musuvathi-CAV08,Sen2011}
provide testing procedures based on enumerating TSO runs and checking that they do not produce happens-before cycles.
Clearly, while these procedures can establish non-robustness, they can never prove a program robust.
On the other hand, \cite{AlglaveM11} provides a sound over-approximate static analysis that allows to prove robustness, but may also inaccurately conclude to non-robustness and insert fences unnecessarily.
We are interested here in developing an approach that allows for precise checking of trace-robustness, and for optimal fence insertion (in a sense defined later).

In our previous work \cite{BMM11}, trace-robustness against TSO has been proved to  be decidable and \pspace-complete, even for unbounded store buffers, in the case of a fixed number of threads and assuming a fixed number of shared variables, ranging over a finite data domain.
The method that shows this decidability and complexity result does not provide a practical algorithm: it is based on a non-deterministic, bounded enumeration of computations. 
Moreover, it does not carry over to the general setting we consider here.
Therefore, in this paper we propose a novel approach to checking robustness that is fundamentally different from \cite{BMM11}.
We provide a general, source-to-source reduction of the trace-robustness problem against TSO to the {\em state reachability problem under SC semantics}. In other words, we show that trace-robustness is not more expensive than SC state reachability, which is the unavoidable problem to be solved by any precise decision algorithm for concurrent programs.
This is the {\em key contribution of the paper} from which we derive other results, such as decidability results in particular cases, as well as an algorithm for efficient fence insertion.

To establish our reduction, we first provide a complete characterization of non-robustness in terms of so-called {\em feasible attacks}. An attack is a pair of load and store instructions of a thread, called the attacker,
whose reordering can lead to a non-SC computation.
In that case we say the attack is feasible, because it has a (TSO) witness computation.
The special form of witness computations then allows us to detect them by tracking SC computations of an {\em instrumented program}.
Given a potential attack, we show how to check its feasibility by solving an SC state reachability query in a linear-size instrumented program.
The fact that {\em only SC semantics} (of the instrumented program) needs to be considered for detecting non-SC behaviors (of the original program) is important: it relieves us from examining TSO computations, which obliges to encode (somehow) the contents of store buffers (as in, e.g., \cite{burckhardt-musuvathi-CAV08,Sen2011}).
Interestingly, the detection of feasible attacks can be parallelized, which speeds up the decision procedure.
Overall, we provide a general reduction of the TSO robustness problem to a quadratic number (in the size of the program) of state reachability queries under the SC semantics in linear-size instrumented programs of the same type as the original one. Our construction is source-to-source and is valid for (1) an arbitrary number of memory addresses/variables, (2) an arbitrary data domain, (3) an arbitrary number of threads, and (4) unbounded store buffers.

With this reduction, we can harness all available techniques and tools for solving reachability queries (either exactly, or approximately) in various classes of concurrent programs, regardless of decidability and complexity issues. It also yields decision algorithms for significant classes of programs. Assume we have a finite number of memory addresses, and the data domain is finite. Then, for a fixed number of threads, a direct consequence of our reduction is that the robustness problem is decidable and in \pspace{} since it is polynomially reducible to state reachability in finite-state concurrent programs \cite{Kozen77}. Therefore, we obtain in this case a {\em simple} robustness checking algorithm which matches the complexity upper bound proved in \cite{BMM11}.
Our construction also provides an effective decision algorithm for the up to now open case where the {\em number of threads is arbitrarily large}. Indeed, state reachability queries in this case can be solved in vector addition systems with states (VASS), or equivalently as coverability in Petri nets, which is known to be decidable \cite{Rackoff1978} and EX\pspace-hard \cite{Lipton1976}. In both cases (fixed or arbitrary number of threads) the decision algorithms do not assume bounded store buffers.

As last contribution, we address the issue of enforcing robustness by fence insertion.
Obviously, inserting a fence after each store ensures robustness, but it also ruins all performance benefits that a relaxed memory model brings.
A natural requirement on the set of fences is irreducibility, i.e.,  minimality wrt. set inclusion.
Since  there may be several irreducible sets enforcing robustness, it is natural to ask for a set that optimizes some notion of cost. We assume that we have a {\em cost function} that defines the cost of inserting a fence at each program location. For instance, by assuming cost $1$ for all locations, we optimize the size of the fence set. Other cost functions reflect the performance of the resulting program.
We propose an algorithm which, given a cost function, computes an optimal set of fences. The algorithm is based on 0/1-integer linear programming and
exploits the notion of attacks to guide the selection of fences.

We implemented the algorithms (using SPIN as a back-end reachability checker), and applied them successfully to several examples, including mutual exclusion protocols and concurrent data structures. The experiments we have carried out show that our approach is quite effective: (1) Many of the attacks to be examined can be discarded by simple syntactic checks (e.g., the presence of a fence between the store and load instructions), and those that require solving reachability queries are handled in few seconds, (2) the fence insertion procedure finds efficiently optimal sets of fences, in particular, it can handle the version of the Non-Blocking Write protocol  \cite{KopetzNBW93} described in \cite{Owens2010} (where the write is guarded by a Linux x86 spinlock) for which Owens' method based on so-called {\em triangular data races} (see related work below) inserts unnecessary fences.\\[2mm]
\paragraph{Related work:}
There are only few results on decidability and complexity of relaxed memory models. Reachability under TSO has been shown to
be decidable but non-primitive recursive \cite{ABBM10} in the case of a finite number of threads and a finite data domain. In the same case, trace-robustness has been shown to be \pspace-complete in \cite{BMM11} using a combinatorial approach.
The method we adopt in this paper is conceptually and technically different from the one we took in \cite{BMM11}. While we reuse from \cite{BMM11} the fact that it is possible to reason on 
TSO computations where only one thread has reordered its actions, we develop here a new approach where the main technical contributions reside in
the characterization of non-robustness in terms of existence of feasible attacks and in
the instrumentation we provide to reduce trace-robustness to SC state reachability. Besides getting decidability and complexity results, this reduction allows to
leverage all the existing verification methods and tools for checking (SC) state reachability in various classes of programs to tackle the issue of checking and enforcing robustness against TSO.

Alur et al. have shown that checking sequential consistency of a concurrent implementation wrt. a specification is undecidable in general \cite{Alur96}.
This result does not contradict our findings: we consider here the special case where the implementation is the TSO semantics and the specification is the SC semantics of a program. In \cite{GibbonsKorach1997}, the problem of deciding whether a given computation is SC feasible has been proved \np-complete. Robustness is concerned with all TSO computations, instead.

As mentioned above, the problem of checking and enforcing trace-robustness against weak memory models has been addressed in \cite{burckhardt-musuvathi-CAV08,Sen2011,AlglaveM11}, but none of these works provide (sound and complete) decision procedures. Owens proposes in \cite{Owens2010} a notion of robustness that is {\em stronger} than trace-robustness, based on detecting triangular data races.
This allows for sound trace-robustness checking but, as Owens shows in his paper, in some cases leads to unnecessary fences which can be harmful for performance.
Moreover, the notion of triangular data races defined in \cite{Owens2010} comes without an algorithm for checking it\footnote{Citation from \cite{Owens2010}: ``{\em ... formal reasoning directly on traces can
be tedious, so a program logic or sound static analyzer specialized to proving
triangular-race freedom might make the application of TRF more convenient.}''}.
Complexity considerations (using the techniques in \cite{BMM11}) show that detecting triangular data races requires solving state reachability queries under SC. Therefore, as we show here, checking trace-robustness is not more expensive than detecting triangular data races.

State-based robustness (which preserves the reachable states) is a weaker robustness criterion, but does not preserve path properties in contrast to trace-robustness, and is of significantly higher complexity (non-primitive recursive as it can be deduced from \cite{ABBM10},  instead of \pspace). It has been addressed in a precise manner in \cite{AbdullaACLR12} where a symbolic decision procedure together with a fence insertion algorithm are provided. The same issue is addressed in \cite{KupersteinVY11,KupersteinVY12} using over-approximate reachability analysis based on abstractions of the store buffers.

Finally, let us mention work that considers the dual approach: starting from a robust program, remove unnecessary fences \cite{VafeiadisN11}.
This work is aimed at compiler optimisations and does not provide a decision procedure for robustness.
It can also not find an optimal set of fences the enforce trace-robustness.

\section{Parallel Programs}\label{Section:Programs}

\paragraph{Syntax}
We consider parallel programs with shared memory as defined by the grammar in Figure~\ref{Figure:SyntaxThreads}.
Programs have a name and consist of a finite number of threads.
Each thread has an identifier and a list of local registers it operates on.
The thread's source code is given as a finite sequence of labelled instructions.
More than one instruction can be marked by the same label;
this allows us to mimic expressive constructs like conditional
branching and iteration with a lightweight syntax.
The instruction set includes loads from memory to a local register, stores to memory,
memory fences to control TSO store buffers, local computations, and assertions.
Figure~\ref{Figure:Dekker} provides a sample program.

\begin{figure}

\begin{minipage}[t]{0.5\textwidth}
\setlength{\grammarindent}{5em}
\setlength{\grammarparsep}{\parskip}
\begin{grammar}
  <prog> ::= "program" <pid> <thrd>$^*$

  <thrd> ::= "thread" <tid>\\"regs" <reg>$^*$\\"init" <label>\\"begin" <linst>$^*$ "end"

  <linst> ::= <label>":" <inst>";" "goto" <label>";"

  <inst> ::= <reg> $\leftarrow$ "mem["<expr>"]"
  \alt "mem["<expr>"]" $\leftarrow$ <expr>
  \alt "mfence"
  \alt <reg> $\leftarrow$ <expr>
  \alt "assert" <expr>

  <expr> ::= <fun>"("<reg>$^*$")"
\end{grammar}
\captionof{figure}{Syntax of parallel programs.}
\label{Figure:SyntaxThreads}
\end{minipage}
\hfill
\begin{minipage}[t]{0.45\textwidth}
\small
\begin{alltt}
\theprogram{Dekker}{
\thethread{\athread_1}{r_1}{l_0}{
\thetransition{l_0}{l_1}{\thestore{\xaddr}{1}}
\thetransition{l_1}{l_2}{\theload{r_1}{\yaddr}}
}
\thethread{\athread_2}{r_2}{l_0'}{
\thetransition{l_0'}{l_1'}{\thestore{\yaddr}{1}}
\thetransition{l_1'}{l_2'}{\theload{r_2}{\xaddr}}
}
}
\end{alltt}
\captionof{figure}{Simplified version of Dekker's mutex algorithm. Under SC, it is impossible that $r_1=r_2=0$ when both threads reach $l_2$ and $l_2'$.}
\label{Figure:Dekker}
\end{minipage}

\end{figure}

We assume a program comes with two sets: a \emph{data domain} $\datadomain$ and a \emph{function domain} $\fundom$.
The data domain should contain value zero: $0\in\datadomain$.
Moreover, we assume that all values from $\datadomain$ can be used as addresses.
Each memory location accessed by loads and stores is identified by such an address, and memory locations identified by different addresses do not overlap.
The set $\fundom$ contains functions that are defined on the data domain and
can be used in the program.
Note that we do not make any finiteness assumptions.
\paragraph{TSO Semantics}
Fix a program $\aprog$ with threads $\threaddom:=\{t_1, \ldots, t_n\}$.
Let each thread $\athread_i$ have the initial label $\alab_{0, i}$ and declare registers $\listof{\areg_i}$.
We define the set of variables as the union of addresses and registers: $\vardom:= \datadomain \cup \cup_{i\in [1, n]}\listof{\areg_i}$.
We denote the set of all instruction labels that occur in threads by $\labdom$.

The TSO semantics is operational, in terms of states and labelled transitions between them. On the x86 TSO architecture, each processor effectively has a local FIFO buffer that keeps stores for later execution~\cite{Owens2009,SewellCACM2010,burckhardt-musuvathi-CAV08,Sen2011}.
Therefore, a \emph{state} is a triple $\conf=(\pcconf, \valconf, \bufconf)$ where program counter $\pcconf\colon\threaddom\rightarrow\labdom$ holds, for each thread, the label of the instruction(s) to be executed next.
The valuation $\valconf:\vardom\rightarrow\datadomain$ gives the values of registers and memory locations.
The third component $\bufconf\colon\threaddom\rightarrow(\datadomain\times\datadomain)^*$ is the (per thread) buffer content: a sequence of address-value pairs $\assignl{\anaddr}{\aval}$.

In the \emph{initial state} $\initconf := (\initpcconf, \initvalconf, \initbufconf)$
the program counter is set to the initial labels, $\initpcconf(t_i)=\alab_{0, i}$ for all $t_i\in \threaddom$,
registers and addresses hold value zero, $\initvalconf(x)=0$ for all $x\in\vardom$,
and all buffers are empty, $\initbufconf(\athread):=\varepsilon$ for all $t\in\threaddom$.
Here, $\varepsilon$ denotes the empty sequence.

Instructions yield transitions between states that are labelled by \emph{actions} from $\actdom:=\threaddom\times\left(\set{\issueact, \localact}\cup(\set{\loadact,\storeact}\times\datadomain\times\datadomain)\right)$.
An action consists of the thread id and the actual arguments of an executed instruction.
We use $\localact$ to abstract assignments and asserts that are local to the thread.
The issue action $\issueact$ indicates that a store was executed on the processor.
The store action $(\athread,\storeact,\anaddr,\aval)$ gives the moment when the store becomes visible in memory.

The \emph{TSO transition relation} $\tsotrans{}$ is the smallest relation between TSO states that is defined by the rules in Table~\ref{Table:TSORules}.
The rules repeat, up to notation and support for locked instructions, Figure~1 from \cite{Owens2009}.
The first two rules implement loads from the buffer and from the memory respectively.
By the third rule, store instructions enqueue write operations to the buffer.
The fourth rule non-deterministically dequeues and executes them on memory.
The fifth rule defines that memory fences can only be executed when the buffer is empty.
The last two rules refer to local assignments and assertions.
We omitted locked instructions to keep things simple.
Their introduction is straightforward and does not affect the results.
Indeed, our implementation supports them \cite{Trencher}.

\begin{table}[thb!]
\centering
\therules{
\therule
{\synt{instr} = \theload{\areg}{\afun_\anaddr(\listof{\areg_\anaddr})}, $\anaddr = \afun_\anaddr(\valconf(\listof{\areg_\anaddr}))$, $\projection{\bufconf(\athread)}{(\assignl{\anaddr}{*})} = \beta\cdot(\assignl{\anaddr}{\aval})$}
{$(\pcconf,\valconf,\bufconf)\tsotrans{(\athread,\loadact,\anaddr,\aval)} (\pcconf',\valconf[\areg:=\aval],\bufconf)$}
\therule
{\synt{instr} = \theload{\areg}{\afun_\anaddr(\listof{\areg_\anaddr})}, $\anaddr = \afun_\anaddr(\valconf(\listof{\areg_\anaddr}))$, $\projection{\bufconf(\athread)}{(\assignl{\anaddr}{*})} = \varepsilon$, $\aval=\valconf(\anaddr)$}
{$(\pcconf,\valconf,\bufconf)\tsotrans{(\athread,\loadact,\anaddr,\aval)} (\pcconf',\valconf[\areg:=\aval],\bufconf)$}
\therule
{\synt{instr} = \thestore{\afun_\anaddr(\listof{\areg_\anaddr})}{\afun_\aval(\listof{\areg_\aval})}, $\anaddr = \afun_\anaddr(\valconf(\listof{\areg_\anaddr}))$, $\aval=\afun_\aval(\valconf(\listof{\areg_\aval}))$}
{$(\pcconf,\valconf,\bufconf)\tsotrans{(\athread,\issueact)}(\pcconf',\valconf,\bufconf[\athread:=\bufconf(\athread)\cdot(\assignl{\anaddr}{\aval})])$}
\therule
{$\bufconf(\athread) = (\assignl{\anaddr}{\aval})\cdot\beta$}
{$(\pcconf,\valconf,\bufconf)\tsotrans{(\athread,\storeact,\anaddr,\aval)}(\pcconf,\valconf[\anaddr:=\aval],\bufconf[\athread:=\beta])$}
\therule
{\synt{instr} = \themfence{}, $\bufconf(\athread)=\varepsilon$}
{$(\pcconf,\valconf,\bufconf)\tsotrans{(\athread,\localact)}(\pcconf',\valconf,\bufconf)$}
\therule
{\synt{instr} = \thelocal{\areg}{\afun(\listof{\areg})}}
{$(\pcconf,\valconf,\bufconf)\tsotrans{(\athread,\localact)}(\pcconf',\valconf[\areg:=\afun(\valconf(\listof{\areg}))],\bufconf)$}
\therule
{\synt{instr} = \thecondition{\afun(\listof{\areg})}, $\afun(\valconf(\listof{\areg}))\neq 0$}
{$(\pcconf,\valconf,\bufconf)\tsotrans{(\athread,\localact)}(\pcconf',\valconf,\bufconf)$}
}
\caption{TSO transition rules, assuming $\pcconf(\athread)=\alab$, an instruction \synt{instr} at label $\alab$ with destination $\alab'$, and $\pcconf':=\pcconf[\athread:=\alab']$.
We use $\projection{}{}$ to denote projection and $*$ for any value, i.e., $\projection{\bufconf(\athread)}{(\assignl{\anaddr}{*})}$ is a list of address-value pairs in the buffer of thread $\athread$ having the address $a$.}
\label{Table:TSORules}
\end{table}

The set of \emph{TSO computations} contains all sequences of actions that lead from the initial TSO state to a state where all buffers are empty:
\begin{align*}
 \tsocomputof{\aprog}&:=\setcond{\tau\in\actdom^*}{\initconf\tsotrans{\tau}
\conf\text{ for some TSO state }\phantom{\}}\\
&\hspace{3cm}\conf=(\pcconf, \valconf, \bufconf)\text{ with
}\bufconf(\athread)=\varepsilon\text{ for all }\athread\in\threaddom}.
\end{align*}
The requirement of empty buffers is not important for our results but rather a modelling choice.
Figure~\ref{Figure:DekkerTsoComputation} presents a TSO computation of Dekker's program where the store of the first thread is delayed past the load.

\begin{figure}[htb!]
\centering
\begin{tikzpicture}
  \node(tau){$\tau =$};
  \node(issue1)[right of=tau,node distance=0.8cm,color=red]{$(\athread_1,\issueact)$};
  \node(load1)[right of=issue1,node distance=1.3cm,color=red]{$(\athread_1,\loadact,\yaddr,0)$};
  \node(issue2)[right of=load1,node distance=1.3cm,color=blue]{$(\athread_2,\issueact)$};
  \node(store2)[right of=issue2,node distance=1.3cm,color=blue]{$(\athread_2,\storeact,\yaddr,1)$};
  \node(load2)[right of=store2,node distance=1.55cm,color=blue]{$(\athread_2,\loadact,\xaddr,0)$};
  \node(store1)[right of=load2,node distance=1.55cm,color=red]{$(\athread_1,\storeact,\xaddr,1)$};
  \draw[-,thick] (issue1.north)edge[out=8,in=172](store1.north);
\end{tikzpicture}
\caption{A TSO computation of Dekker's algorithm.
Actions drawn in red belong to the first thread, actions in blue belong to the second thread.
The arc connects the issue action with the corresponding delayed store action of the first thread.}
\label{Figure:DekkerTsoComputation}
\end{figure}

\paragraph{SC Semantics}
Under SC \cite{Lamport79}, stores are not buffered and hence states are pairs $(\pcconf, \valconf)$.
The rules for SC transitions are appropriately simplified TSO rules.
To avoid case distinctions between TSO and SC in the definition of traces, a store instruction generates two actions: an issue followed by the store.
Memory fences have no effect under SC.
We denote the set of all SC computations of $\aprog$ by
\begin{align*}
\sccomputof{\aprog}:=\setcond{\sigma\in
\actdom^*}{\initconf\sctrans{\sigma}\conf\text{ for some SC state }\conf}.
\end{align*}

\section{TSO Robustness}\label{Section:Robustness}
Robustness ensures that the behaviour of a program does not change when it is run on TSO hardware as compared to SC.
We study trace-based robustness as in \cite{ShashaSnir88,burckhardt-musuvathi-CAV08,Sen2011,AlglaveM11,BMM11}.
Traces capture the essence of a computation: the control and data dependencies among actions.
More formally, consider some computation $\alpha\in\sccomputof{\aprog}\cup\tsocomputof{\aprog}$.
The \emph{trace} $\traceof{\alpha}$ is a graph where the nodes are labelled by the actions in $\alpha$ (stores and issue yield one node).
The arcs are defined by the following relations.
We have a per thread $\athread\in\threaddom$ total order $\progorder^{\athread}$ that gives the order in which the actions of $\athread$ where issued.
Similarly, we have a per address $a\in\datadomain$ total order $\storeorder^{a}$ that gives the ordering of all stores to $a$.
We call the unions $\progorder\ := \cup_{\athread\in\threaddom}\progorder^{\athread}$ and $\storeorder\ :=\cup_{a\in\datadomain}\storeorder^{a}$
the \emph{program order} and the \emph{store order} of the trace.
Finally, there is a \emph{source relation} $\sourceorder$ that determines the store from which a load receives its value.
By $\traces_{\textsf{mm}}(\aprog):=\traceof{\comput_{\textsf{mm}}(\aprog)}$ with $\textsf{mm}\in\{\SC, \TSO\}$ we denote the set of all \emph{SC/TSO traces} of program $\aprog$.
The \emph{TSO robustness problem} checks whether the sets coincide.
\vspace{2mm}
\begin{quote}
{\bf Given:} A parallel program $\aprog$.\\
{\bf Problem:} Does $\tsotracesof{\aprog}=\sctracesof{\aprog}$ hold?
\end{quote}
\vspace{2mm}
Since inclusion $\sctracesof{\aprog}\subseteq\tsotracesof{\aprog}$ always holds, we only have to check the reverse inclusion.
We call a computation $\tau\in\tsocomputof{\aprog}$ \emph{violating} if its trace is not among the SC traces of the program, i.e., $\traceof{\tau}\notin\sctracesof{\aprog}$.
Violating \TSO-computations employ cyclic accesses to addresses that SC is unable to serialize \cite{ShashaSnir88}.
The cyclic accesses are made visible using a \emph{conflict relation} from loads to stores.
Intuitively, $\loadact\conflictorder \storeact$ means that $\storeact$ overwrites a value that $\loadact$ reads.
The union of all four relations is commonly called \emph{happens-before relation} of the trace, $\happensbefore\ :=\ \progorder\cup\storeorder\cup\sourceorder\cup\conflictorder$.
\begin{lemma}[{\cite{ShashaSnir88}}]
\label{Lemma:CharacterisationViolationHappensBefore}
Consider TSO trace $\traceof{\tau}\in\tsotracesof{\aprog}$.
Then $\traceof{\tau}\in \sctracesof{\aprog}$ iff $\happensbefore$ is acyclic.
\end{lemma}
Consider the computation in Figure~\ref{Figure:DekkerTsoComputation}.
The load from thread $\athread_1$ conflicts with the store from $\athread_2$ because the load reads the initial value of $\yaddr$ while the store overwrites it.
The situation with the load from $\athread_2$ and the store from $\athread_1$ is symmetric.
Together with the program order, the conflict relations produce a cycle:
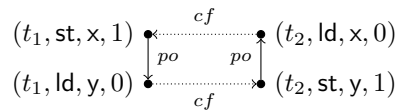
\begin{minipage}{\textwidth}
\centering
\begin{tikzpicture}[node distance=0.66cm,every state/.style={inner sep=1pt,fill=black!100,minimum size=0.5mm,initial text=}]
    \node[state] (sx1) [label=left:{$(\athread_1,\storeact,\xaddr,1)$}] {};
    \node[state] (ly)  [below of=sx1,label=left:{$(\athread_1,\loadact,\yaddr,0)$}] {};

    \node[state] (sy1) [node distance=1.5cm,right of=ly,label=right:{$(\athread_2,\storeact,\yaddr,1)$}] {};
    \node[state] (lx)  [above of=sy1,label=right:{$(\athread_2,\loadact,\xaddr,0)$}] {};

    \draw[->] (sx1) -- (ly) node [midway,right] {$\scriptstyle po$};

    \draw[->] (sy1) -- (lx) node [midway,left] {$\scriptstyle po$};

    \draw[->] (ly) edge[densely dotted] node[midway,below] {$\scriptstyle cf$} (sy1);
    \draw[->] (lx) edge[densely dotted] node[midway,above] {$\scriptstyle cf$} (sx1);
\end{tikzpicture}

\end{minipage}
Indeed, there is no SC computation with this trace, as predicted by Lemma~\ref{Lemma:CharacterisationViolationHappensBefore}.

Lemma~\ref{Lemma:CharacterisationViolationHappensBefore} does not provide a method for finding cyclic traces.
We have recently shown that TSO robustness is decidable, in fact, \pspace-complete \cite{BMM11}.
The algorithm underlying this result, however, is based on enumeration and not meant to be implemented.
The main contribution of the present work is a novel and practical approach to checking robustness.

The only concept we keep from our earlier work are minimal violations.
A \emph{minimal violation} is a violating computation that uses a minimal total number of delays.
Interestingly, for minimal violations the following holds.
\begin{lemma}[Locality \cite{BMM11}, Appendix~\ref{Appendix:locality}]\label{Lemma:Locality}
In a minimal violation, only a single thread delays stores.
\end{lemma}
Consider the computation in Figure~\ref{Figure:DekkerTsoComputation}.
It relies on a single delay in thread $\athread_1$ and, indeed, is a minimal violation.
As predicted by the lemma, the second thread writes to its buffer and immediately flushes it.

\section{Attacks on TSO Robustness}\label{Section:Attacks}
Our approach to checking TSO robustness combines two insights.
We first rephrase robustness in terms of a simpler problem: the absence of
feasible attacks.
We then devise an algorithm that checks attacks for feasibility.
Interestingly, SC reachability techniques are sufficient for this purpose.
Together, this yields a sound and complete reduction of TSO robustness to SC reachability.

The notion of attacks is inspired by the shape of minimal violations.
We show that if a program is not robust, then there are violations of the form shown in Figure~\ref{Figure:ShapeViolation}:
one thread, the \emph{attacker}, delays a store action $\attackstore$ past a later load action $\attackload$ in order to break robustness.
The remaining threads become \emph{helpers} and provide a happens-before path from $\attackload$ to $\attackstore$.
This yields a happens-before cycle and shows non-robustness.

Thread, store instruction $\storeinst$ of $\attackstore$, and load instruction $\loadinst$ of $\attackload$ are syntactic objects.
The idea of our approach is to fix these three parameters, the \emph{attack}, prior to the analysis.
The algorithm then tries to find a witness computation that proves the attack feasible.
\begin{figure}[t]
\centering
\begin{tikzpicture}[node distance=1.8cm]

  \node(inv){$\tau =$ };
  \node(isu)[right of=inv,node distance=1.3cm,color=red]{$\issueact_{\attackstore}$};
  \node(ld)[right of=isu,node distance=1.8cm,color=red] {$\attackload$};
  \node(st) [right of=ld,color=red] {$\attackstore$};
  \node(b1) [right of=isu,node distance=1.2cm] {};
  \node(b2) [right of=st,node distance=2.5cm] {};

  \draw[-, very thick] (inv)edge node[below]{$\tau_1$} (isu);
  \draw[-, very thick] (isu)edge node[below]{$\tau_2$} (ld);
  \draw[-, very thick,color=blue] (ld)edge node[below,color=black]{$\tau_3$}(st);
  \draw[-, very thick,style=dashed,color=red] (st)edge node[below,color=black]{$\tau_4$}(b2);
  \draw[-,thick] (isu.north east)edge[out=25,in=155](st.north west);
  \draw[-,thick] (b1.north)edge[out=20,in=160](b2.north);
\end{tikzpicture}
\caption{TSO witness for the attack $(\attacker, \storeinst, \loadinst)$.
It satisfies the following constraints. \wita\ Only the attacker delays stores. \witb\ Store $\attackstore$ is an instance of $\storeinst$. It is the first store of the attacker that is delayed. Load $\attackload$ is an instance of $\loadinst$. It is the last action of the attacker that is overstepped by $\attackstore$. So $\tau_2$ contains loads, assignments, asserts, and issues, but no fences and stores of the attacker. It may contain arbitrary helper actions.
\witc\ We require $\attackload\happensbefore^+\anact$ for every action $\anact$ in $\attackload\cdot\tau_3\cdot\attackstore$. An issue + store of a helper is counted as one action $\anact$.
\witd\ Sequence $\tau_4$ only consists of stores of the attacker that were issued before $\attackload$ and that have been delayed.
\wite\ All these stores $\storeact$ satisfy $\adrof{\storeact}\neq \adrof{\attackload}$, i.e., $\attackload$ has not read its value early.}
\label{Figure:ShapeViolation}
\end{figure}

\begin{definition}\label{Definition:Attacks}
An \emph{attack} $\attack=(\attacker,\storeinst,\loadinst)$ consists of a thread $\attacker\in\threaddom$ called \emph{attacker}, a store instruction $\storeinst$ and a load instruction $\loadinst$.
A \emph{TSO witness for $A$} is a computation of the form in  Figure~\ref{Figure:ShapeViolation}, i.e., it satisfies \wita\ to \wite.
If a TSO witness exists, the attack is called \emph{feasible}.
\end{definition}
In Dekker's algorithm, there is an attack $\attack=(\attacker, \storeinst, \loadinst)$
with $\attacker = \athread_1$, $\storeinst$ the store at $l_0$, and $\loadinst$ the load at $l_1$.
A TSO witness of this attack is the computation $\tau$ from Figure~\ref{Figure:DekkerTsoComputation}. With reference to Figure~\ref{Figure:ShapeViolation} we have
$\tau_1=\varepsilon$,
$\issueact_{\attackstore}=(\athread_1,\issueact)$,
$\tau_2=\varepsilon$,
$\attackload=(\athread_1,\loadact,\yaddr,0)$,
$\tau_3=(\athread_2,\issueact)\cdot(\athread_2,\storeact,\yaddr,1)\cdot(\athread_2,\loadact,\xaddr,0)$,
$\attackstore=(\athread_1,\storeact,\xaddr,1)$,
$\tau_4=\varepsilon$.
The program also contains a symmetric attack $\attack'$ with $\athread_2$ as the attacker.

Although TSO witnesses are quite restrictive computations, robustness can be reduced to verifying that no attack has a TSO witness. 
\begin{theorem}[Complete Characterization of Robustness with Attacks]\label{Theorem:CharacterizationRobustness}
Program $\aprog$ is robust iff no attack is feasible, i.e., no attack admits a TSO witness.
\end{theorem}
\begin{proof}
The existence of a TSO witness implies non-robustness of the program.
Indeed, a TSO witness comes with a happens-before cycle $\attackstore\progorder^+\attackload\happensbefore^+\attackstore$.
We argue that also the reverse holds: if a program is not robust, there is a feasible attack.
Assume $\aprog$ is not robust.
We construct a TSO witness computation.
Among the violating computations, we select $\tau\in\tsocomputof{\aprog}$ where the number of delays is minimal.
The computation need not be unique.
By Lemma~\ref{Lemma:Locality}, in $\tau$ only one thread $\attacker$ uses its buffer and \wita\ holds.
We elaborate on the shape of $\tau$.

Initially, the attacker executes under SC so that stores
immediately follow their issues.
This computation is embedded into $\tau_1$ in Figure~\ref{Figure:ShapeViolation}.
Eventually, the attacker starts delaying stores.
Let $\attackstore$ be the first store that is delayed.
It gets reordered past several loads, the last of which being $\attackload$.
This shows \witb.

The helper actions in $\tau_3$ are depicted in blue in
Figure~\ref{Figure:ShapeViolation}.
To see that we can assume \witc, first note that $\attackload\happensbefore^+ \attackstore$ holds.
If there was no such path, $\attackstore$ could be placed before $\attackload$ without changing
the trace.
This would save a delay, in contradiction to minimality of $\tau$.
Assume $\tau_3=\tau_3'\cdot \anact\cdot \tau_3''$ where $\anact$ is the first action so that $\attackload\not\happensbefore^+\anact$.
Then $\anact$ is independent from all actions in $\attackload\cdot\tau_3'$.
We find a computation with the same trace where $\anact$ is placed before $\attackload$.

With cycle $\attackstore\progorder^+\attackload\happensbefore^+\attackstore$,  computation $\tau_4$ only needs to contain the stores of the attacker that have been delayed past $\attackload$.
Since these stores are non-blocking, the helpers can stop with the last action in $\tau_3$.
We can moreover assume $\attackload$ to be the program
order last action of the attacker.
\witd\ holds.

We now argue that $\attackload$ has not read its value early from any of the delayed stores, \wite.
Towards a contradiction, assume $\attackload$ obtained its value from $\storeact$ in $\tau_4=\tau_{41}\cdot \storeact\cdot \tau_{42}$.
There is a computation $\tau'$ where we avoid the early read: it replaces $\tau_4$ by $\tau_{41}\cdot \storeact\cdot\attackload\cdot \tau_{42}$.
The traces of $\tau$ and $\tau'$ coincide, but $\tau'$ saves the delay of $\storeact$ past $\attackload$.
A contradiction to minimality.

It is readily checked that $\tau$ is a TSO witness for the attack $(\attacker, \storeinst, \loadinst)$ where $\storeinst$ and $\loadinst$ are the instructions that
$\attackstore$ and $\attackload$ are derived from. \qed
\end{proof}
Since the number of attacks is only quadratic in the size of the program, we
can just enumerate them and check whether one admits a TSO
witness.
To check whether a witness exists, we employ the instrumentation described in the
following section.

\section{Instrumentation}\label{Section:Reduction}
Consider program $\aprog$ with attack $\attack = (\attacker, \storeinst, \loadinst)$.
TSO witnesses for $\attack$ only make limited use of buffers, to an extent that allows us to characterize them by SC computations in a program $\aprog_\attack$ that is \emph{instrumented for attack $\attack$}.
By instrumentation we mean that $\aprog_{\attack}$ replaces every thread by a modified version.
Capturing TSO witnesses with a program that runs under SC is difficult for two reasons.
First, TSO has unbounded store buffers which can delay stores arbitrarily long.
Second, the happens-before dependence that the helpers create may involve an arbitrary number of actions.
Our instrumentation copes with both problems using the following tricks.

To handle store buffering, we instrument the attacker thread (Section~\ref{Section:Attacker}).
Essentially, we emulate store buffering under SC using auxiliary addresses.
To explain the idea, consider the TSO witness in Figure~\ref{Figure:ShapeViolation}.
When the instrumented attacker executes the delayed stores  $\attackstore\cdot \tau_4$ under SC, they occur right behind their issue actions.
To mimic store buffering, these stores now access auxiliary addresses that the other threads do not load.
As a result, the stores remain invisible to the helpers.
This is as intended: the delayed stores $\attackstore\cdot \tau_4$ in Figure~\ref{Figure:ShapeViolation} are also never accessed by helper threads.
But how many auxiliary addresses do we need to faithfully simulate buffers?
It is sufficient to have \emph{a single auxiliary address} per address in the program.
The reason is that a load always reads the most recent store to its address that is held in the buffer.

To build up a happens-before path from $\attackload$ to $\attackstore$, we instrument the helper threads (Section~\ref{Section:Helpers}).
The question is how to decide whether a new action $\anact$ is in happens-before relation with an earlier action $\anact'$ so that $\attackload\happensbefore^*\anact'\happensbefore^*\anact$.
What is the information we need about the earlier actions in order to append $\anact$?
It is sufficient to know two facts.
Has the thread already contributed an action $\anact'$?
This information ensures $\anact'\progorder^*\anact$, and can be kept in the control flow of the thread.
Moreover, we keep track of whether the path contains a load or store access to the address $\adrof{\anact}$.
If there was a load access $\anact'=\loadact$, we can add a store $\anact=\storeact$ and get $\loadact\happensbefore^*\storeact$.
If there was a store, we are free to add a load or a store.
Hence, we need \emph{one auxiliary address} per address in the program for this access information: no access, load access, store access.

Consider the TSO witness for Dekker given in Figure~\ref{Figure:DekkerTsoComputation}.
Instead of buffering $\redtext{(\athread_1,\storeact,\xaddr,1)}$, the instrumentation immediately executes the store after its issue action.
But instead of address $\xaddr$, the action accesses the auxiliary address $(\xaddr, \delay)$ that the other threads do not load.
To indicate that this store is invisible to the helper threads, we depict it in gray.
So, the SC computation of the instrumented program roughly looks like this:
\begin{align*}
\graytext{(\athread_1,\issueact)}
\cdot \graytext{(\athread_1,\storeact,(\xaddr,d),1)}
\cdot \redtext{(\athread_1,\loadact,\yaddr,0)}
\overset{(1)}{\cdot} \bluetext{(\athread_2,\issueact)}
\bluetext{(\athread_2,\storeact,\yaddr,1)}
\overset{(2)}{\cdot} \bluetext{(\athread_2,\loadact,\xaddr,0)}.
\end{align*}
At moment (1), we know that there has been a load access to address $\yaddr$.
At moment (2), address $\yaddr$ has even seen a store.
At the end of the computation, address $\yaddr$ has seen a store and address $\xaddr$ has seen a load.

The store of $\athread_2$ can be appended since it is in happens-before relation with the attacker's load.
The following load can be added as $\athread_2$ has contributed the previous store.
The search terminates here since the helper's load accesses address $\xaddr$ that was used by the store from the attack.
\subsection{Instrumentation of the Attacker}\label{Section:Attacker}
The instrumentation emulates the buffering of stores in a TSO witness (Figure~\ref{Figure:ShapeViolation}).
Starting from $\attackstore$, the stores are replaced by stores $\attacklocalstore$ to auxiliary addresses $(\anaddr, \delay)$ that are only visible to the attacker.
As long as $\anaddr$ has not been written, $(\anaddr, \delay)$ holds the initial value $0$.
Once the attacker stores $\aval$ into $\anaddr$, we set $\themem{(\anaddr, \delay)}=(\aval, \delay)$.
In this way, $(\anaddr, \delay)$ always holds the most recent store.
A load $\theload{\areg}{\anaddr}$ of the attacker reads a value $\aval$ from the buffer whenever $\themem{(\anaddr, \delay)}=(\aval, \delay)$; otherwise $\themem{(\anaddr, \delay)}=0$ and the load obtains the value $\aval=\themem{\anaddr}$ from memory.
We turn to the translation.

Let $\attacker$ declare registers $\areg^*$,
have initial location $\alab_{0}$, and define instructions
$\langle linst\rangle^*$ that contain $\storeinst$ and $\loadinst$ from the attack.
The instrumentation is
\begin{align*}
\sem{\attacker}:=\thethreadnoinst{\tilde\attacker&}{\areg^*}{\ \alab_0\
}\\
&\theinstructions{\
  \langle linst\rangle^*\ \semattackermove{\storeinst}\ \semattackermove{\loadinst}\ \semattacker{\langle linst\rangle}^*\
}.
\end{align*}
It introduces a copy of the source code $\semattacker{\langle linst\rangle}^*$ where the stores are replaced by accesses to auxiliary addresses.
To move to the code copy, the attacker uses an instrumented version
$\semattackermove{\storeinst}$ of $\storeinst$.
\begin{figure}[ht]
\begin{eqnarray}
\semattackermove{\thetransition{\alab_1}{\alab_2}{\thestore{e_1}{e_2}}} &:=&\thetransition{\alab_1}{\tilde\alab_{x}}{\thestore{(e_1,\delay)}{(e_2,\delay)}}\label{Equation:StoreInst}\\
&&\thetransition{\tilde\alab_{x}}{\tilde\alab_2}{\thestore{a_{\attackstore}}{e_1}}\notag\\[1mm]
\semattackermove{\thetransition{\alab_1}{\alab_2}{\theload{\areg}{\anexpr}}}&:=&
\thetransition{\tilde\alab_{1}}{\tilde\alab_{x1}}{\thecondition{\themem{(e,\delay)} = 0}}\label{Equation:LoadInst}\\
&&\thetransition{\tilde\alab_{x1}}{\tilde\alab_{x2}}{\thestore{\hb}{\mytrue}}\notag\\
&&\thetransition{\tilde\alab_{x2}}{\tilde\alab_{x3}}{\thestore{(\anexpr,\hb)}{\loadacc}}\notag\\[1mm]
\semattacker{\thetransition{\alab_1}{\alab_2}{\thestore{e_1}{ e_2}}} &:=& \thetransition{\tilde\alab_1}{\tilde\alab_2}{\thestore{(e_1,\delay)}{(e_2,\delay)}}\label{Equation:Store}\\[1mm]
\semattacker{\thetransition{\alab_1}{\alab_2}{\theload{\areg}{\anexpr}}} &:=&
\thetransition{\tilde\alab_{1}}{\tilde\alab_{x1}}{\thecondition{\themem{(\anexpr,\delay)} = 0}}\label{Equation:Load}\\
&&\thetransition{\tilde\alab_{x1}}{\tilde\alab_{2}}{\theload{\areg}{\anexpr}}\notag\\
&&\thetransition{\tilde\alab_{1}}{\tilde\alab_{x2}}{\thecondition{\themem{(\anexpr,\delay)} \neq 0}}\notag\\
&&\thetransition{\tilde\alab_{x2}}{\tilde\alab_{2}}{\theload{(\areg,\delay)}{(\anexpr,\delay)}}\notag\\[1mm]
\semattacker{\thetransition{\alab_1}{\alab_2}{\mathit{local}}}&:=& \thetransition{\tilde\alab_1}{\tilde\alab_2}{\mathit{local}}\label{Equation:Local}\\[1mm]
\semattacker{\thetransition{\alab_1}{\alab_2}{\themfence}} &:=&\label{Equation:Fence}
\end{eqnarray}
\caption{Instrumentation of the attacker.}
\label{Figure:TranslationAttacker}
\end{figure}

The translation of instructions is defined in Figure~\ref{Figure:TranslationAttacker}.
We make a few remarks.
The instrumentation of $\storeinst = \thetransition{\alab_1}{\alab_2}{\thestore{e_1}{e_2}}$
keeps the address used in the store in a fresh address
$\anaddr_{\attackstore}$.
For the sake of readability, in Equation~\eqref{Equation:Load} we use memory accesses in instructions other than load and store.
Equation~\eqref{Equation:Fence} deletes fences, as they forbid to delay $\attackstore$ over $\attackload$.
Let $\loadinst=\thetransition{\alab_1}{\alab_2}{\theload{\areg}{\anexpr}}$ be the load used in the attack.
Equation~\eqref{Equation:LoadInst} checks that the load has not read its value early and sets an auxiliary happens-before address $(\anexpr,\hb)$ to access level load, $\loadacc$.
We postpone the definition of access levels until the translation of helpers.
It also sets $\hb$ flag for helpers to indicate that they cannot execute actions not contributing to the happens-before path.
Figure~\ref{Figure:DekkerAttackerInstrumentation} illustrates the instrumentation on our running example.
\begin{figure}[t]
\begin{minipage}[t]{0.5\textwidth}
\begin{alltt}
\thethread{\tilde{\athread_1}}{r_1}{l_0}{
/* Original code */
\thetransition{l_0}{l_1}{\thestore{\xaddr}{1}}
\thetransition{l_1}{l_2}{\theload{r_1}{\yaddr}}

/* Instrumented stinst */
\thetransition{l_0}{\tilde\alab_{x}}{\thestore{(\xaddr,\delay)}{(1,\delay)}}
\thetransition{\tilde\alab_{x}}{\tilde\alab_1}{\thestore{a_{\attackstore}}{\xaddr}}

/* Instrumented ldinst */
\thetransition{\tilde\alab_{1}}{\tilde\alab_{x1}}{\thecondition{\themem{(\yaddr,\delay)} = 0}}
\thetransition{\tilde\alab_{x1}}{\tilde\alab_{x2}}{\thestore{\hb}{\mytrue}}
\thetransition{\tilde\alab_{x2}}{\tilde\alab_{x3}}{\thestore{(\yaddr,\hb)}{\loadacc}}
}
\end{alltt}
\end{minipage}
\begin{minipage}[t]{0.5\textwidth}
\begin{alltt}

/* Instrumented copy of the store */
\thetransition{\tilde\alab_0}{\tilde\alab_1}{\thestore{(\xaddr,\delay)}{(1,\delay)}}

/* Instrumented copy of the load */
\thetransition{\tilde\alab_{1}}{\tilde\alab_{x4}}{\thecondition{\themem{(\yaddr,\delay)} = 0}}
\thetransition{\tilde\alab_{x4}}{\tilde\alab_{2}}{\theload{\areg}{\yaddr}}
\thetransition{\tilde\alab_{1}}{\tilde\alab_{x5}}{\thecondition{\themem{(\yaddr,\delay)} \neq 0}}
\thetransition{\tilde\alab_{x5}}{\tilde\alab_{2}}{\theload{(\areg,\delay)}{(\yaddr,\delay)}}

\end{alltt}
\end{minipage}

\caption{Attacker instrumentation of thread $\athread_1$ in Dekker from Figure~\ref{Figure:Dekker}.
The attack's store is the store at label $l_0$, the load is the load at label $l_1$.}
\label{Figure:DekkerAttackerInstrumentation}
\end{figure}

\subsection{Instrumentation of Helpers}\label{Section:Helpers}
In TSO witnesses, by \witc, all helper actions after $\attackload$ are in happens-before relation with $\attackload$.
To ensure this, we make use of Lemma~\ref{Lemma:Access}.
The proof from left to right is by definition of happens before.
For the reverse direction, note that happens-before is \emph{stable under insertion}.
Consider $\storeact\sourceorder \loadact$.
A happens-before relation remains valid in any computation that places actions between $\storeact$ and $\loadact$.
\begin{lemma}\label{Lemma:Access}
Consider $\tau = \tau_1\cdot\anact_1\cdot\tau_2\in\sccomputof{\aprog}$ where for all $\anact_2$ in $\tau_2$ we have $\anact_1\happensbefore^* \anact_2$.
Computation $\tau\cdot \anact$ satisfies $\anact_1 \happensbefore^* \anact$ iff
\setlist{nolistsep}
\begin{description}
\item[(i)] there is an action $\anact_2$ in $\anact_1\cdot \tau_2$ with $\threadof{\anact_2}=\threadof{\anact}$ or
\item[(ii)] $\anact$ is a load whose address is stored in $\anact_1\cdot \tau_2$ or
\item[(iii)] $\anact$ is a store (with issue) whose address is loaded or stored in $\anact_1\cdot \tau_2$.
\end{description}
\setlist{listsep}
\end{lemma}
The lemma suggests the following instrumentation.
For every helper $\athread$, we track whether it has executed an action that depends on $\attackload$.
The idea is to use the control flow.
Upon detection of this first action, the thread moves to a copy of its code.
All actions from this copy stay in happens-before relation with $\attackload$.

It remains to decide whether an action $\anact$ allows a thread to move to the code copy.
According to Lemma~\ref{Lemma:Access}, this depends on the earlier accesses to the address $\anaddr=\adrof{\anact}$.
We introduce auxiliary \emph{happens-before addresses} $(\anaddr, \hb)$ that provide this access information.
The addresses $(\anaddr, \hb)$ range over the domain $\{\noacc, \loadacc, \storeacc\}$ of \emph{access types}.
It is sufficient to keep track of the \emph{maximal} access type wrt. the  ordering\quad $\noacc \text{ (no access)}  <\ \loadacc\ \text{ (load access)} <\ \storeacc \text{ (store access)}$.

For the definition, consider a helper thread $\athread$ that declares $\areg^*$, has initial label $\alab_0$, and defines instructions $\langle linst\rangle^*$.
The instrumented thread is
\begin{align*}
\sem{\athread}:=\thethreadnoinst{\ \tilde\athread\ &}{\ \tilde \areg, \areg^*\ }{\ \alab_0\ }\\
&\theinstructions{\ \semhelperorig{\langle{}linst\rangle^}*\  \semhelpertrans{\langle ldstinst\rangle}^*\ \semhelpercpy{\langle linst\rangle}^*\ \semhelpersuc{\langle\alab\rangle}^*\ }.
\end{align*}
Here, $\langle ldstinst\rangle^*$ is the subsequence of all load and store instructions.
Their instrumentation $\semhelpertrans{\langle ldstinst\rangle}^*$ is used to move to the code copy $\semhelpercpy{\langle linst\rangle}^*$.
Moreover, $\langle\alab\rangle^*$ are all labels used by the thread.
The additional instructions $\semhelpersuc{\langle\alab\rangle}^*$ raise a success flag when a TSO witness has been found.
$\semhelperorig{\langle linst \rangle}$ forces helpers to either enter the code copy or stop when $\hb$ flag is raised.

\begin{figure}[t]
\begin{eqnarray}
\semhelperorig{\thetransition{\alab_1}{\alab_2}{\mathit{instr}}}&:=&
\thetransition{\alab_1}{\alab_{x}}{\thecondition{\themem{\hb} = 0}}\label{Equation:HelperOriginal}\\
&&\thetransition{\alab_x}{\alab_2}{\mathit{instr}}\notag\\
\semhelpertrans{\thetransition{\alab_1}{\alab_2}{\theload{\areg}{\anexpr}}}&:=&
\thetransition{\alab_1}{\tilde \alab_{x}}{\thecondition{\themem{(\anexpr,\hb)}=\storeacc}}\label{Equation:HelperLoadMove}\\
&&\thetransition{\tilde \alab_x}{\tilde \alab_2}{\theload{\areg}{\anexpr}}\notag\\[1mm]
\semhelpertrans{\thetransition{\alab_1}{\alab_2}{\thestore{\anexpr_1}{\anexpr_2}}}&:=&
\thetransition{\alab_1}{\tilde \alab_{x1}}{\thecondition{\themem{(\anexpr,\hb))}\geq \loadacc}}\label{Equation:HelperStoreMove}\\
&&\thetransition{\tilde \alab_{x1}}{\tilde \alab_{x2}}{\thestore{\anexpr_1}{\anexpr_2}}\notag\\
&&\thetransition{\tilde \alab_{x2}}{\tilde \alab_{2}}{\thestore{(\anexpr_1,\hb)}{\storeacc}}\notag\\[1mm]
\semhelpercpy{\thetransition{\alab_1}{\alab_2}{\textit{local/mfence}}} &:=& \thetransition{\tilde \alab_1}{\tilde \alab_2}{\textit{local/mfence}}\label{Equation:HelperLocalFence}\\[1mm]
\semhelpercpy{\thetransition{\alab_1}{\alab_2}{\thestore{e_1}{e_2}}}&:=&
\thetransition{\tilde \alab_{1}}{\tilde \alab_{e}}{\thestore{ e_1}{ e_2}}\label{Equation:HelperStore}\\
&&\thetransition{\tilde \alab_{e}}{\tilde \alab_{2}}{\thestore{(e_1,\hb)}{\storeacc}}\notag\\[1mm]
\semhelpercpy{\thetransition{\alab_1}{\alab_2}{\theload{\areg}{\anexpr}}}&:=&
\thetransition{\tilde \alab_1}{\tilde \alab_{x1}}{\thelocal{\tilde \areg}{\anexpr}}\label{Equation:HelperLoad}\\
&&\thetransition{\tilde\alab_{x1}}{\tilde \alab_{x2}}{\theload{\areg}{\tilde\areg}}\notag\\
&&\thetransition{\tilde\alab_{x2}}{\tilde \alab_{2}}{\thestore{(\tilde\areg,\hb)}{\maxfun\{\loadacc, \themem{(\tilde\areg,\hb)}\}}}\notag\\
\semhelpersuc{\alab}&:=&
\thetransition{\tilde\alab}{\tilde\alab_{x1}}{\theload{\tilde r}{a_{\attackstore}}}\label{Equation:HelperSuccessCheck}\\
&&\thetransition{\tilde\alab_{x1}}{\tilde\alab_{x2}}{\theload{\tilde r}{(\tilde r,\hb)}}\notag\\
&&\thetransition{\tilde\alab_{x2}}{\tilde\alab_{x3}}{\thecondition{\tilde r \neq 0}}\notag\\
&&\thetransition{\tilde\alab_{x3}}{\tilde \alab_{x4}}{\thestore{\successvar}{\mytrue}}\notag
\end{eqnarray}
\caption{Instrumentation of helpers.\label{Figure:HelpersInstrumentation}}
\end{figure}

The translation of instructions is given in Figure~\ref{Figure:HelpersInstrumentation}.
We make some remarks.
Transitions to the code copy check the auxiliary addresses for whether the current action is in happens-before relation with $\attackload$.
Loads in Equation~\eqref{Equation:HelperLoadMove} check for an earlier store access to their address, Lemma~\ref{Lemma:Access}(ii).
Stores in Equation~\eqref{Equation:HelperStoreMove} require that the address has seen at least a load, Lemma~\ref{Lemma:Access}(iii).
They set the access level to \storeacc{}.
Loads and stores in the code copy maintain the auxiliary addresses to contain the maximal access types, Equations~\eqref{Equation:HelperLoad} and~\eqref{Equation:HelperStore}.
Note the auxiliary register $\tilde \areg$ that ensures we do not overwrite the address.
At every label of the code copy we add a check, Equation \eqref{Equation:HelperSuccessCheck}, whether the address used in the attack's store has been accessed in the code copy.
If so, a success flag is raised.

\subsection{Soundness and Completeness}
The flag indicates that the SC computation corresponds to a TSO witness,
and we call $(\pcconf, \valconf)$ with $\valconf(\successvar)=\mytrue$ a
\emph{goal configuration}.
The instrumentation thus reduces feasibility of attack $\attack$ to SC reachability of a goal configuration in program $\aprog_{\attack}$.
The instrumentation is sound and complete.
If a goal configuration is reachable, we can reconstruct a TSO witness for
the attack.
In turn, every TSO witness ensures the goal configuration is reachable.
\begin{theorem}[Soundness and Completeness]
Attack $\attack=(\attacker,\storeinst,\loadinst)$ is feasible in program $\aprog$ iff program $\aprog_{\attack}$ reaches a goal configuration under SC.
\end{theorem}
In combination with Theorem~\ref{Theorem:CharacterizationRobustness}, we
can check robustness by inspecting all $\aprog_{\attack}$.
\begin{theorem}[From Robustness to SC Reachability]\label{Theorem:RobReach}
Program $\aprog$ is robust iff no instrumentation $\aprog_{\attack}$
reaches a goal configuration under SC.
\end{theorem}
The instrumentation we provide is linear in size.
Then, it follows from Theorem~\ref{Theorem:RobReach} that checking robustness for programs over finite data domains is in \pspace.
The problem is actually \pspace-complete due to the lower
bound in \cite{BMM11}.

\section{Robustness for Parameterized Programs}\label{Section:Parameterization}
We extend the study of robustness to \emph{parameterized programs}.
A parameterized program represents an infinite family of instance programs that replicate the threads multiple times.
Syntactically, parameterized programs coincide with the parallel programs we introduced in Section~\ref{Section:Programs}: they have a name and declare a finite set of threads $\athread_1,\ldots, \athread_k$.
The difference is in the semantics.
A parameterized program represents a family of programs: for every vector $I=(n_1,\ldots, n_k)\in \nat^k$, a \emph{program instance} $\aprog(I)$ declares $n_i$ copies of thread $\athread_i$.

In the parameterized setting, the robustness problem asks whether all instances of a given program are robust:
\begin{quote}
{\bf Given:} A parameterized program $\aprog$.\\
{\bf Problem:} Does $\tsotracesof{\aprog(I)}=\sctracesof{\aprog(I)}$ hold for all instances $\aprog(I)$ of $\aprog$?
\end{quote}
The problem is interesting because libraries usually cannot make assumptions on the number of threads that use their functions.
They have to guarantee proper functioning for any number.

We reduce robustness for parameterized programs to a parameterized version of reachability, based on the following insight.
A parameterized program is not robust if and only if there is an instance $\aprog(I)$ that is not robust.
With Theorem~\ref{Theorem:CharacterizationRobustness}, instance $\aprog(I)$ is not robust if and only if there is an attack $\attack$ that is feasible.
With the instrumentation from Section~\ref{Section:Reduction} and Theorem~\ref{Theorem:RobReach}, this feasibility can be checked as reachability of a goal configuration in $\aprog(I)_{\attack}$.

Algorithmically, it is impossible to instrument all (infinitely many) instance programs.
Instead, the idea is to instrument directly the parameterized program $\aprog$ towards the attack $\attack$.
Using the constructions from Section~\ref{Section:Reduction}, we modify every thread and again obtain program $\aprog_{\attack}$, which is now parameterized.

Actually, for the attacker we have to be slightly more careful. 
In an instance program, only one copy of the thread should act as attacker, the remaining copies have to behave like helpers.
Therefore, the thread must be instrumented not only as an attacker, but also as a helper. 
To ensure that only one copy of the attacker delays stores, we add an additional flag variable. 
Before starting an attack, the thread checks this variable. 
If it contains the initial value, the thread sets the flag and starts delaying stores. If it has a different value, the thread continues to run sequentially. 
This check requires an atomic test-and-set operation which can be implemented on x86 by the \texttt{lock cmpxchg} instruction. 
Support for locked instructions is immediate to add to our programming model.

Modulo these two changes, the instances $\aprog_{\attack}(I)$ coincide with the instrumentations $\aprog(I)_{\attack}$.
Together with the argumentation in last two paragraphs this justifies the following theorem.
\begin{theorem}\label{Theorem:Parameter}
A parameterized program $\aprog$ is not robust iff there is an attack $\attack$ so that an instance $\aprog_{\attack}(I)$ of
program $\aprog_{\attack}$ reaches a goal configuration under SC.
\end{theorem}
Reachability of a goal configuration in one instance of $\aprog_{\attack}$ can be reformulated as a coverability problem for Petri nets, which is known to be decidable \cite{Rackoff1978}.
The key observation in the reduction to Petri nets is that threads in instance programs never use their identifiers, simply because they are copies of the same source code.
This means there is no need to track the identity of threads, it is sufficient to count how many instances of a thread are in each state --- a technique known as counter abstraction \cite{GermanSistla1992}.
\begin{theorem}\label{Theorem:ParameterizedRobustnessIsDecidableRef}
Robustness for parameterized programs over finite data domains is decidable and \expspace-hard --- already for Boolean programs.
\end{theorem}
For the lower bound, we in turn encode the coverability problem for Petri nets into robustness for parameterized programs \cite{Trencher,Lipton1976}

\section{Fence Insertion}\label{Section:FenceInsertion}
To ease the presentation, we return to parallel programs.
Since the algorithm only relies on a robustness checker, it carries over to the parametric setting.

Our goal is to insert a set of fences that ensure robustness of the resulting program.
By \emph{inserting a fence at label $\alab$} we mean the following modification of the program.
Introduce a fresh label $\alab_f$.
Then, translate each instruction \thetransition{\alab}{\alab'}{inst} into \thetransition{\alab_f}{\alab'}{inst}.
Finally, add an instruction \thetransition{\alab}{\alab_f}{\themfence}.

We call a set of labels $\afenceset$ in program $\aprog$ a \emph{valid fence set} if inserting fences at these labels yields a robust program.
We say that $\afenceset$ is \emph{irreducible} if no strict subset is a valid fence set.
In general, however, we would like to compute a valid fence set which is \emph{optimal} in some sense.
We pose the \emph{fence computation problem}:
\begin{quote}
{\bf Given:} A program $\aprog$ and a strictly positive \emph{cost function} $\costfun\colon\labdom\to\realnums^+$.\\
{\bf Problem:} Compute a valid fence set $\afenceset$ with $\Sigma_{\alab\in\afenceset}\costfun(\alab)$ minimal.
\end{quote}
Since we assume $\costfun$ to be strictly positive, every optimal fence set is irreducible.

We consider two criteria of optimality: minimization of program size and maximization of program performance.
By solving the problem for $\costfun\equiv 1$ we compute a fence set of minimal size, thus minimizing the code size of the fenced program.
Maximization of program performance requires minimizing the number of times memory fence instructions are executed: practical measurements~\cite{Trencher} show that it is impossible to save CPU cycles by executing more fences, but with less stores in the TSO buffer.
For this, $\costfun(\alab)$ is defined as the frequency at which instructions labeled by $\alab$ occur in executions of the original program $\aprog$.
Concrete values of $\costfun$ can be either estimated by profiling or computed by mathematical reasoning about the program.

From a complexity point of view, fence computation is at least as hard as robustness.
Indeed, robustness holds if and only if the optimal valid fence set is $\afenceset=\emptyset$.
Actually, since fence sets can be enumerated, computing an optimal valid fence set does not require more space than checking robustness.
Notice that this also holds in the parameterized case.
\begin{theorem}\label{Theorem:ComplexityOptimalFenceSet}
For programs over finite domains, fence computation is \pspace-complete.
In the parameterized case, it is decidable and \expspace-hard.
\end{theorem}
In the remainder of the section, we give a practical algorithm for computing optimal valid fence sets.
\subsection{Fence Sets for Attacks}
We say that a label $\alab$ is \emph{involved in the attack} $\attack = (\attacker, \storeinst, \loadinst)$ if it belongs to some path in the control flow graph of $\attacker$ from the destination label of $\storeinst$ to the source label of $\loadinst$.
We denote the set of all such labels
by $\alabelset_{\attack}$.

We call a set of labels $\afenceset_\attack$ \emph{an eliminating fence set for attack $\attack$} if adding fences at all labels in $\afenceset_\attack$ eliminates the attack.
Dekker's algorithm has two eliminating fence sets: $\afenceset_{\attack} = \set{l_1}$ eliminates the only attack by $\athread_1$, and $\afenceset_{\attack'} = \set{l_1'}$ eliminates the only attack by $\athread_2$.
Actually, the sets are \emph{irreducible}: no strict subset eliminates the attack.
Note that any irreducible eliminating set $\afenceset_\attack$ satisfies $\afenceset_\attack\subseteq\alabelset_{\attack}$.

\begin{lemma}\label{Lemma:FenceSetAsUnion}
Every irreducible valid fence set $\afenceset$ can be represented as a union of irreducible eliminating fence sets for all feasible attacks.
\end{lemma}
\begin{proof}
By Theorem~\ref{Theorem:CharacterizationRobustness}, fence set $\afenceset$ eliminates all feasible attacks.
Therefore, it includes some irreducible eliminating fence set $\afenceset_{\attack}$ for every feasible attack $\attack$.
By irreducibility, $\afenceset$ cannot contain labels outside the union of these $\afenceset_{\attack}$ sets.
\qed
\end{proof}
In compliance with the lemma, in the Dekker's program $\afenceset=\afenceset_{\attack}\cup\afenceset_{\attack'}$.

Lemma~\ref{Lemma:FenceSetAsUnion} is useful for fence computation since optimal fence sets are always irreducible.
All irreducible eliminating fence sets for attacks can be constructed by an exhaustive search through all selections of labels involved in the attack.
For each candidate fence set, to judge whether it eliminates the attack, we check SC reachability in the instrumented program as described in Section~\ref{Section:Reduction}.

Note that this search may raise an exponential number of reachability queries.
In practice this rarely constitutes a problem.
First, attacks seldom have a large number of involved labels, so the number of candidates is small.
Second, the reachability checks can be avoided if a candidate fence set covers all the ways in the control flow graph from $\storeinst$ to $\loadinst$.

\subsection{Computing an Optimal Valid Fence Set}
To choose among the sets $\afenceset_{\attack}$, we set up a $0/1$-integer linear programming (ILP) problem $M_{\aprog}\matmult x_{\aprog}\geq b_{\aprog}$.
The optimal solutions $f(x_{\aprog})\rightarrow \min$
correspond to optimal fence sets.
Here, $0/1$ means the variables are restricted to yield Booleans.

We define inequalities that encode the feasible attacks with their corrections.
Consider attack $\attack$ for which we have determined the irreducible eliminating fence sets $\afenceset_{1},\ldots, \afenceset_n$.
For each set, we introduce a variable $x_{\afenceset_i}$  and set up Inequality~\eqref{Equation:ILP}(left).
It selects a fence set to eliminate the attack.
\begin{alignat}{3}
\sum_{1\leq i\leq n}x_{\afenceset_i}&\geq 1 \hspace{1cm}
\sum_{\alab\in \afenceset_i}x_{\alab}&\geq  \power{\afenceset_i}x_{\afenceset_i}\hspace{1cm}
 f(x_{\aprog})&:=\sum_{\alab\in \labdom}\costfun(\alab)x_\alab.
\label{Equation:ILP}
\end{alignat}
When $\afenceset_i$ has been chosen, we insert a fence at each of its labels $\alab$.
We add further variables $x_\alab$, and encode this insertion by Inequality~\eqref{Equation:ILP}(center).
By definition of the ILP, the variables $x_{\afenceset_i}$ and $x_\alab$
will only take Boolean values $0$ or $1$.
So if $x_{\afenceset_i}$ is set to $1$, the inequality requires that all $x_\alab$ with $\alab\in \afenceset_i$ are set to $1$.

Our goal is to select fences with minimal costs.
We encode this into the objective function \eqref{Equation:ILP}(right).
An optimal solution $x^*$ of the resulting 0/1-ILP denotes the fence set $\afenceset(x^*):=\setcond{\alab\in \labdom}{x^*_\alab=1}$.

\begin{theorem}
$\afenceset(x^*)$ is valid and optimal, and thus solves fence computation.
\end{theorem}

\section{Experimental Evaluation}\label{Section:Experiments}
We implemented our algorithms
in a prototype
called \trencher{}~\cite{Trencher}.
The tool performs the reduction of robustness to SC reachability given in  Section~\ref{Section:Reduction} and computes a minimal fence set as described in Section~\ref{Section:FenceInsertion}.
\trencher{} executes independent reachability queries in parallel and uses Spin~\cite{Holzmann97} as back-end model checker.
With \trencher{}, we have performed a series of experiments.

\subsection{Examples}
The first class of examples are mutual exclusion protocols that are implemented via shared variables.
These protocols are typically not robust under TSO and require additional fences after stores to synchronization variables.
We studied robust and non-robust instances of Dekker and Peterson for two threads, as well as Lamport's fast mutex \cite{lamport87:fast} for three threads.
Moreover, we checked the CLH and MCS Locks, robust list-based queue locks that use compare-and-set \cite{Herlihy2008}.

As second class of examples, we considered concurrent data structures.
The Lock-Free Stack is a concurrent stack implementation using compare-and-swap~\cite{Herlihy2008}.
Cilk's THE WSQ is a work stealing queue from the implementation of the Cilk-5 programming language~\cite{Frigo1998}.

Finally, we consider miscellaneous concurrent algorithms that are known to be sensitive to program order relaxations.
We analyse several instances of the Non-Blocking Write protocol \cite{KopetzNBW93}.
NBWL is the spinlock + non-blocking write example considered by Owens in Section 8 of \cite{Owens2010}.
Finally, our tool discovers the known bug in Java's Parker implementation that is due to TSO relaxations~\cite{dice09:park}.

The test inputs are available online~\cite{Trencher}.

\begin{table}[thb!]
\centering
\small
\setlength{\tabcolsep}{2.75pt}
\begin{tabular}{|l||r|r|r||r|r|r|r||r||r|r|r|}
\hline
Program&T&L&I&RQ&NR1&NR2&R&F&Spin&Ver&Real\\
\hline
\hline
Peterson (non-robust)&2&14&18&23&2&12&9&2&7.7&0.5&2.9\\
\hline
Peterson (robust)&2&16&20&12&12&0&0&0&0.0&0.0&0.0\\
\hline
Dekker (non-robust)&2&24&30&95&12&28&55&4&31.7&2.1&14.2\\
\hline
Dekker (robust)&2&32&38&30&30&0&0&0&0.0&0.0&0.0\\
\hline
Lamport (non-robust)&3&33&36&36&9&15&12&6&14.4&6.0&5.9\\
\hline
Lamport (robust)&3&39&42&27&27&0&0&0&0.0&0.0&0.0\\
\hline
CLH Lock (robust)&7&62&58&54&48&6&0&0&4.9&0.2&1.6\\
\hline
MCS Lock (robust)&4&52&50&30&26&4&0&0&2.9&0.4&0.9\\
\hline
\hline
Lock-Free Stack (robust)&4&46&50&14&14&0&0&0&0.0&0.0&0.0\\
\hline
Cilk's THE WSQ (non-robust)&5&86&79&152&141&8&3&3&10.0&18.0&7.4\\
\hline
\hline
NBW2 (non-robust)&2&21&19&15&9&5&1&1&2.5&0.2&0.8\\
\hline
NBW3 (robust)&2&22&20&15&15&0&0&0&0.0&0.0&0.0\\
\hline
NBW4 (robust)&3&25&22&9&7&2&0&0&0.7&0.1&0.4\\
\hline
NBWL (robust)&4&45&45&30&26&4&0&0&2.7&0.2&0.7\\
\hline
Parker (non-robust)&2&9&8&2&0&1&1&1&0.5&0.0&0.3\\
\hline
Parker (robust)&2&10&9&2&2&0&0&0&0.0&0.0&0.0\\
\hline
\end{tabular}

\vskip 0.5em
\caption{Benchmarking results.}
\label{Table:Experiments}
\end{table}

\subsection{Benchmarking}

We executed \trencher{} on the examples, using a machine with Intel(R) Core(TM) i5 CPU M 560 @ 2.67GHz (4 cores) running GNU/Linux.
Table~\ref{Table:Experiments} summarizes the results.
The columns T, L, and I give the number of threads, labels, and instructions in the example.
RQ is the number of reachability queries raised by \trencher{}.
Provided the example is robust, this number is equal to the number of attacks $(\attacker, \storeinst, \loadinst)$.
NR1 is the number of verification queries that were answered negatively by \trencher{} itself, without running Spin.
Such queries correspond to attacks where $\storeinst$ cannot be delayed past $\loadinst$ because of memory fences or locked instructions in between.
NR2 and R are the numbers of queries that are answered negatively/positively by the external model checker.
Hence, RQ = NR1 + NR2 + R.
F is the number of fences inserted.

The column Spin gives the total CPU time taken by Spin and Clang, the C compiler, to produce a verifier executable (pan).
The column Ver provides the total CPU time taken by \trencher{} and the external verifier.
Real is the wall-clock time in seconds of processing an example.
All times are given in seconds.

\subsection{Discussion}
The analysis of robust algorithms is particularly fast.
They typically only have a small number of attacks that have to be checked by a model checker.
Robust Dekker and Peterson do not have such attacks at all.
In the CLH and MCS locks, their number is less than 20\%.

In some examples (non-robust Dekker, CLH Lock, NBW2, NBW4), up to 94\% of the CPU time was spent on generating verifiers.
This leaves room for improvement by switching to a model checker without compilation phase.
For some examples (LamNR, CLHLock), the wall-clock time constitutes $1/3$ to $1/4$ of the CPU time (4-cores).
This confirms good parallelizability of the approach.

Remarkably, our trace-based analysis can establish robustness of the NBWL example, as opposed to the earlier analyses via triangular data races which would have to place a fence \cite{Owens2010}.

\acks
The second author has been granted by the Competence Center High Performance Computing and Visualization (CC-HPC) of the Fraunhofer Institute for Industrial Mathematics (ITWM).

\bibliographystyle{plain}
\bibliography{cited}

\clearpage
\appendix
\section{Definition of Traces}
Since the definition of traces in Section~\ref{Section:Robustness} was a bit brief, we recall here the full construction.
Consider an SC or a TSO computation $\alpha\in\sccomputof{\aprog}\cup\tsocomputof{\aprog}$.
Its trace $\traceof{\alpha}$ is a node-labelled graph $(N, \lambda, \progorder, \storeorder, \sourceorder)$ with nodes $N$, labelling $\lambda: N\rightarrow \actdom$, and $\progorder, \storeorder, \sourceorder\ \subseteq N\times N$ the aforementioned relations that define the edges.
The \emph{program order} is a union of $\progorder\ = \cup_{\athread\in\threaddom}\progorderof{\athread}$ of per thread total orders.
The \emph{store ordering} $\storeorder\ = \cup_{\anaddr\in\datadomain}\storeorderof{\anaddr}$ gives a total order for the stores to each address.
We use the syntax $\maxof{\progorderof{\athread}}$ and $\maxof{\storeorderof{\anaddr}}$ to access the maximal elements in these total orders.
Finally, we have a \emph{source relation} $\sourceorder$ from stores to loads.

Traces are defined inductively, starting from the empty trace for the empty word $\varepsilon$.
Assume we already constructed $\traceof{\alpha}=(N, \lambda, \progorder, \storeorder, \sourceorder)$.
In the definition of $\traceof{\alpha\cdot\anact}:=(N\cup \{n\}, \lambda', \progorder', \storeorder', \sourceorder')$, the choice of $n$ depends on the type of $\anact$.
If we have a store, we use the moment the action was issued.
Otherwise, we add a new node:
\begin{description}
\item[$\anact=(\athread,\storeact,\anaddr,\aval)$]
Let $n$ be the minimal node in $\progorderof{\athread}$ labelled by $\lambda(n)=\issueact$. 
We set $\lambda':=\lambda[n:=\anact]$ and $\progorder'\ :=\ \progorder$.
\item[$\anact\neq (\athread,\storeact,\anaddr,\aval)$] We add a fresh node $n\notin N$ to the trace, set
    $\lambda':=\lambda\cup\{(n, \anact)\}$, and $\progorder'\ :=\ \progorder\cup\ \{(\maxof{\progorderof{\athread}}, n)\}$.
\end{description}
The store order is updated only for stores $(\athread,\storeact,\anaddr,\aval)$.
We define $\storeorder'\ :=\ \storeorder \cup\ \{(\maxof{\storeorderof{\anaddr}},
n)\}$.
The relation is not changed otherwise.
The source relation is updated only for loads and stores.
In case of a load $(\athread,\loadact,\anaddr,\aval)$ we set $\sourceorder'\ :=\ \sourceorder \cup\ \{(\maxof{\storeorderof{\anaddr}}, n)\}$.
In case of a store $(\athread,\storeact,\anaddr,\aval)$ we update the source relation for loads that read from the store early: for all nodes $m$ with $n\progorder^+m$ and $\lambda(m) = (\athread,\loadact,\anaddr,\aval)$ we set $\sourceorder'\ :=\ (\sourceorder \setminus\ \set{(*, m)}) \cup\ \set{(n, m)}$.\\[0.2cm]
Consider trace $\traceof{\alpha}$ with $\alpha\in\tsocomputof{\aprog}$.
The \emph{conflict relation} $\conflictorder$ from load to store actions makes cyclic accesses in the trace visible.
We define $\loadact\conflictorder\storeact$ if there is another store action $\storeact'$ in $\traceof{\tau}$ that satisfies $\storeact'\sourceorder \loadact$ and $\storeact'\storeorder \storeact$.
If $\loadact$ reads the initial value of an address and $\storeact$ overwrites it, we also have $\loadact\conflictorder \storeact$.
The \emph{happens-before relation} of a trace is a union of all four relations: $\happensbefore\ :=\ \progorder\cup\storeorder\cup\sourceorder\cup\conflictorder$.

\section{Minimal Violations and Locality}
\label{Appendix:locality}
in our earlier work~\cite{BMM11} we showed that in a minimal violation only one thread reorders its actions.
Since we employ here a more elaborate programming model, this locality result has to be checked again.

Consider a computation $\tau=\alpha\cdot a\cdot\beta\cdot b\cdot\gamma\in\tsocomputof{\aprog}$ with two actions $a$ and $b$ of the same thread $\threadof{a}=\athread=\threadof{b}$.
We define the \emph{distance $\thedistance{\tau}{a}{b}$ between $a$ and $b$ in $\tau$} as the number of actions in $\beta$ that also belong to this thread:
$\thedistance{\tau}{a}{b}:=\length{\projectionOf{\beta}{\athread}}$.
The \emph{number of delays} $\thedelays{\tau}$ in computation $\tau$ is the sum of distances between corresponding issue and store actions:
$$
\thedelays{\tau} := \sum_{\text{corr. }\issueact,\storeact\text{ in }\tau}\thedistance{\tau}{\issueact}{\storeact}.
$$

We call a violating computation $\tau$ a \emph{minimal violation} if it is has a minimal number of delays among all violating computations.
Clearly, a program $\aprog$ has violating computations if and only if it has a minimal violation.

The following lemma says that if a store action has been delayed, then it has been delayed past a load action of the same thread.
Moreover, the load did not read the value of this store action early.
\begin{lemma}
\label{Lemma:WritesDelayPastReads}
Consider a minimal violation $\tau = \alpha\cdot\issueact\cdot\beta\cdot\storeact\cdot\gamma\in\tsocomputof{\aprog}$, where $\issueact$ and $\storeact$ stem from the same instruction instance of thread $\athread$.
Then $\projectionOf{\beta}{\athread}$ is either empty, or $\projectionOf{\beta}{\athread}=\beta'\cdot\loadact\cdot\beta''$ where $\loadact$ is a load action with $\adrof{\loadact}\neq\adrof{\storeact}$ and $\beta''$ contains only store actions.
\end{lemma}
\begin{proof}
Suppose $\beta$ contains one or more actions of thread $\athread$.
If all actions of thread $\athread$ in $\beta$ are stores, then also
$\tau'=\alpha\cdot\beta\cdot\issueact\cdot\storeact\cdot\gamma$
is a TSO computation of $\aprog$. It has the same trace as $\tau$ but $\thedelays{\tau'} < \thedelays{\tau}$, which contradicts minimality of $\tau$.

Otherwise let $a$ be the last non-store action in $\projectionOf{\beta}{\athread}$, i.e., $\beta=\beta_1\cdot a\cdot\beta_2$ and all actions in $\beta_2$ are stores or belong to threads different from $\athread$.
Since store actions cannot be delayed past a memory fence of the same thread, $a$ is an issue action, a local action, or a load.
In the former two cases, as well as if $a$ is a load from $\adrof{\loadact}=\adrof{\storeact}$, delaying $\storeact$ past $a$ can be avoided in the computation $\tau' = \alpha\cdot\issueact\cdot\beta_1\cdot\beta_2\cdot\storeact\cdot a\cdot\gamma$ of $\aprog$. It has the same trace as $\tau$ and $\thedelays{\tau'} < \thedelays{\tau}$, which contradicts minimality of $\tau$. \qed
\end{proof}
In the remainder of the section, we develop a method to detect happens-before relations in a trace with the help of embedded computations.
We relate two actions in a computation iff the corresponding nodes in the trace are related.
To avoid case distinctions for issue and store actions that yield the same node in the trace, we introduce the \emph{issue relation} $\issueorder$ that links them: $\issueact\issueorder\storeact$.
We include $\issueorder$ into $\happensbefore$.

\begin{definition}[\cite{BMM11}]\label{Definition:HBThrough}
Let $\tau = \alpha\cdot a\cdot \beta\cdot b\cdot \gamma\in\tsocomputof{\aprog}$.
We say \emph{$a$ happens-before $b$ through $\beta$} if there is a (potentially empty) subsequence $c_1\ldots c_n$ of $\beta$ that satisfies (assuming $c_0:= a$ and $c_{n+1}:= b$):
$$a_i\rightarrow_{\hb} a_{i+1}\quad\text{ or }\quad{}a_i\progorder^+a_{i+1}\quad\text{ for all }i\in[0, n].$$
\end{definition}
The next lemma states that the just defined relation is stable under insertion.
\begin{lemma}[\cite{BMM11}]
\label{Lemma:StabilityUnderInsertion}
Consider computations $\tau=\alpha\cdot a\cdot\beta\cdot b\cdot\gamma$ and $\tau'=\alpha'\cdot a\cdot\beta'\cdot b\cdot\gamma'$ in $\tsocomputof{\aprog}$ so that $\projectionOf{\tau}{\athread}=\projectionOf{\tau'}{\athread}$ for every thread $\athread$.
Let $\beta$ be a subsequence of $\beta'$.
Then if $a\happensbefore^+b$ through $\beta$ then $a\happensbefore^+b$ through $\beta'$.
\end{lemma}
The following lemma says that if two actions in a minimal violation are not related via $\happensbefore^+$, they can be reordered without changing the trace and the order of actions within each thread.
\begin{lemma}[\cite{BMM11}]
\label{Lemma:Duality}
Consider a minimal violation $\tau=\alpha\cdot a\cdot\beta \cdot b\cdot \gamma\in\tsocomputof{\aprog}$.
Then (1) $a\happensbefore^+b$ through $\beta$ or (2) there is $\tau'=\alpha\cdot\beta_1\cdot b\cdot a\cdot\beta_2\cdot\gamma\in\tsocomputof{\aprog}$ so that $\traceof{\tau}=\traceof{\tau'}$ and $\projectionOf{\tau}{\athread}=\projectionOf{\tau'}{\athread}$ for every thread $\athread$.
\end{lemma}
\begin{proof}
We establish $\neg(1)\Rightarrow(2)$. Note that this proves the disjunction since $\neg (2)\Rightarrow (1)$ is the contrapositive.
We proceed by induction on $\length{\beta}$ and slightly strengthen the hypothesis: we also show that $\beta_2$ is a subsequence of $\beta$.\\[0.2cm]
\textbf{Base case:} $\length{\beta} = 0$.
Then $\tau = \alpha\cdot a\cdot b\cdot \gamma$ and $a\not\happensbefore b$.
If $\threadof{a}=\threadof{b}$, then $b\progorder^+ a$.
Therefore, $b$ is a store action which has been delayed past $a$.
Swapping $a$ and $b$ will save the delay without changing the trace, in contradiction to the minimality of $\tau$.

If $\threadof{a}\neq\threadof{b}$, then either at least one of the two actions is local, the actions access different addresses, or both are loads.
In all cases swapping them produces $\tau'$ as required in the statement of the lemma.\\[0.2cm]
\textbf{Step case:} Assume the statement holds for $\length{\beta} \leq n$.
Consider $\tau'=\alpha\cdot a\cdot \beta \cdot b\cdot \gamma$ with $\length{\beta} = n + 1$.
Let $c$ be the last action in $\beta = \beta'\cdot c$.
Since $a\not\happensbefore^+ b$ through $\beta$, then $a\not\happensbefore^+ c$ through $\beta'$ or $c\not\happensbefore b$.

Let $a\not\happensbefore^+ c$.
We apply the induction hypothesis to $\tau$ with respect to $a$ and $c$.  This gives $\tau' = \alpha\cdot\beta_1'\cdot c\cdot a\cdot\beta_2'\cdot b\cdot\gamma$ with the same trace and thread computations as $\tau$.
Then, taking into account Lemma~\ref{Lemma:StabilityUnderInsertion}, we apply the hypothesis to $\tau'$ with respect to $a$ and $b$. This yields $\tau'' = \alpha\cdot\beta_1'\cdot c\cdot\beta_{21}'\cdot b\cdot a\cdot\beta_{22}'\cdot\gamma$ having the same trace and thread computations as $\tau'$ and $\tau$.
Note that $\beta_{22}'$ is a subsequence of $\beta_2'$, which in turn is a subsequence of $\beta'$ and hence of $\beta$.

Let $c\not\happensbefore b$. We apply the induction hypothesis to $\tau$ with respect to $b$ and $c$, getting $\tau' = \alpha\cdot a\cdot\beta'\cdot b\cdot c\cdot\gamma$ with the same trace and thread computations as $\tau$.
Applying it again to $\tau'$ with respect to $a$ and $b$ gives $\tau''=\alpha\cdot\beta_1'\cdot b\cdot a\cdot\beta_2'\cdot c\cdot\gamma$.  The computation has the same trace and thread computations as $\tau'$ and $\tau$.
Since $\beta_2'$ is a subsequence of $\beta'$, $\beta_2'\cdot c$ is a subsequence of $\beta$. \qed
\end{proof}

\begin{lemma}[Locality \cite{BMM11}]
In a minimal violation, only a single thread delays stores.
\end{lemma}
\begin{proof}
Consider a minimal violation $\tau\in\tsocomputof{\aprog}$ and suppose at least two threads delayed stores.
By Lemma~\ref{Lemma:WritesDelayPastReads}, each store was delayed past a load of the same thread.
Let $\storeact_2$ of thread $\athread_2$ be the overall last delayed store in $\tau$, and let $\loadact_2$ be the last load of $\athread_2$ overstepped by $\storeact_2$.
Similarly, let $\storeact_1$ be the overall last delayed store in a thread $\athread_1\neq \athread_2$. 
Let $\loadact_1$ be the last load overstepped by $\storeact_1$.

The following fundamental mutual dispositions of reorderings are possible:
\begin{enumerate}
\item $\tau=\gamma_1\cdot\issueact_1\cdot\gamma_2\cdot\loadact_1\cdot\gamma_3\cdot\storeact_1\cdot\gamma_4\cdot\issueact_2\cdot\gamma_5\cdot\loadact_2\cdot\gamma_6\cdot\storeact_2\cdot\gamma_7$
\item $\tau=\gamma_1\cdot\issueact_1\cdot\gamma_2\cdot\loadact_1\cdot\gamma_3\cdot\issueact_2\cdot\gamma_4\cdot\loadact_2\cdot\gamma_5\cdot\storeact_2\cdot\gamma_6\cdot\storeact_1\cdot\gamma_7$
\item $\tau=\gamma_1\cdot\issueact_1\cdot\gamma_2\cdot\loadact_1\cdot\gamma_3\cdot\issueact_2\cdot\gamma_4\cdot\loadact_2\cdot\gamma_5\cdot\storeact_1\cdot\gamma_6\cdot\storeact_2\cdot\gamma_7$
\end{enumerate}
In these three computations every pair $(\loadact_i,\storeact_i)$ provides a happens-before cycle: $\storeact_i\progorder^+\loadact_i$ and, by Lemma~\ref{Lemma:Duality} and minimality, $\loadact_i\happensbefore^+\storeact_i$ through the appropriate subrange of $\tau$.

In the first disposition $\tau$ is not minimal, since it can be shortened to the violating computation $\tau' = \gamma_1\cdot\issueact_1\cdot\gamma_2\cdot\loadact_1\cdot\gamma_3\cdot\storeact_1\cdot\beta$ with $\thedelays{\tau'} < \thedelays{\tau}$. 
Here, $\beta$ contains only store actions of $\athread_2$ that complete earlier issue actions.

In the second disposition $\tau$ is not minimal either.
Starting from $\loadact_1$, thread $\athread_1$ does not perform any actions, except delayed stores, until $\storeact_1$ (Lemma~\ref{Lemma:WritesDelayPastReads}).
Therefore, $\loadact_1$ and all program order later actions of $\athread_1$ can be safely removed from $\tau$ without affecting the happens-before cycle produced by $\athread_2$.
The resulting computation has a smaller number of delays (due to the removed $\loadact_1$), but its trace still includes the cycle by $\athread_2$.
A contradiction to minimality of $\tau$.

Lastly, in the third case $\tau$ is also not minimal.
First we delete $\gamma_7$. Then we erase all actions from $\gamma_6$ that do not belong to $\athread_2$: $\gamma_6'=\projectionOf{\gamma_6}{\athread_2}$. 
By construction, the resulting computation $\tau'$ is a feasible TSO computation:
$$\tau'=\gamma_1\cdot\issueact_1\cdot\gamma_2\cdot\loadact_1\cdot\gamma_3\cdot\issueact_2\cdot\gamma_4\cdot\loadact_2\cdot\gamma_5\cdot\storeact_1\cdot\gamma_6'\cdot\storeact_2.$$
Computation $\tau'$ still contains the happens-before cycle $\storeact_1\progorder^+\loadact_1\happensbefore^+\storeact_1$ inherited from $\tau$.
Since deleting actions cannot increase the number of delays, $\thedelays{\tau'}=\thedelays{\tau}$.
Moreover, since $\tau$ is a minimal violation, so is $\tau'$.

By Lemma~\ref{Lemma:Duality}, $\loadact_2\happensbefore^+\storeact_2$ through $\gamma_5\cdot\storeact_1\cdot\gamma_6'$.
By the choice of $\loadact_1$ and $\storeact_1$ and in accordance with Lemma~\ref{Lemma:WritesDelayPastReads}, $\projectionOf{(\gamma_3\cdot\issueact_2\cdot\gamma_4\cdot\loadact_2\cdot\gamma_5)}{\athread_1}$ only contains delayed stores that were issued before $\loadact_1$.
By definition, $\gamma_6'$ does not contain actions of $\athread_1$ at all.
Therefore, $\loadact_1$ is the program order last action of $\athread_1$. It can be safely removed from $\tau'$ without affecting the cycle of $\athread_2$.
The resulting computation is
$$\tau''=\gamma_1\cdot\issueact_1\cdot\gamma_2\cdot\gamma_3\cdot\issueact_2\cdot\gamma_4\cdot\loadact_2\cdot\gamma_5\cdot\storeact_1\cdot\gamma_6'\cdot\storeact_2.$$
Note that $\thedelays{\tau''} < \thedelays{\tau'}=\thedelays{\tau}$, but computation $\tau''$ still contains the cycle $\storeact_2\progorder^+\loadact_2\happensbefore^+\storeact_2$.
A contradiction to minimality of $\tau$. \qed
\end{proof}

\section{Soundness and Completeness of the Instrumentation}
\begin{theorem}[Soundness and Completeness]
Attack $\attack=(\attacker,\storeinst,\loadinst)$ is feasible in program $\aprog$ iff program $\aprog_{\attack}$ reaches a goal configuration under SC.
\end{theorem}
\begin{proof}
\textbf{Soundness.}
Suppose the instrumented program reaches a goal configuration.
For simplicity, assume that it immediately stops after this.
Then the computation of the instrumented program looks like this:
$$\tau_\attack = \tau_1\cdot\issueact_{\attackstore}\cdot\attacklocalstore\cdot\tau_2\cdot\attackload\cdot\tau_3\cdot\issueact_{\successvar}\cdot\storeact_{\successvar}.$$
The last action, $\storeact_{\successvar}$, is performed by a helper and sets variable $\successvar$ to $\mytrue$, as required by the definition of goal configurations.
This action originates from an instruction generated in accordance with \eqref{Equation:HelperSuccessCheck}.
To reach this instruction, the helper has to enter its code copy.

As required by \eqref{Equation:HelperLoadMove} and \eqref{Equation:HelperStoreMove}, for the first helper to enter its code copy, the attacker must set a $\hb$-variable to a non-zero value by executing $\loadinst$ (action $\attackload$) instrumented in accordance with \eqref{Equation:LoadInst}.
For this, the attacker must enter its code copy and start performing stores to auxiliary addresses.
Accordingly, the first attacker's store to an auxiliary address is denoted by $\attacklocalstore$ in $\tau$. 
It stems from the instrumented $\storeinst$ \eqref{Equation:StoreInst} and is located before $\attackload$.

We elaborate on the contents of $\tau_1$, $\tau_2$, and $\tau_3$.
First, the attacker and helpers execute the code of the original program (helpers — with an additional check at every instruction, \eqref{Equation:HelperOriginal}).
In $\tau_2$ the helpers continue to execute the code of the original program.
Shortly before performing $\attackload$ and stopping, the attacker sets variable $\hb$ to $\mytrue$ thus forcing the helpers to enter their code copies.
Therefore all actions in $\tau_3$ belong to helpers that have entered their code copies.
Also, $\tau_2$ only contains stores of the attacker to auxiliary addresses, and $\tau_3$ does not contain attacker action at all, as follows from \eqref{Equation:LoadInst}.\\[0.2cm]
We now turn $\tau_{\attack}$ into the following TSO witness computation:
$$\tau = \tau_1'\cdot\issueact_{\attackstore}\cdot\tau_2'\cdot\attackload\cdot\tau_3'\cdot\attackstore\cdot\tau_4'.$$
Here, $\tau_1'$ is the subsequence of all $\tau_1$ actions that are produced by instructions from $\aprog$ (this is $\tau_1$ without the conditionals introduced in \eqref{Equation:HelperOriginal}).
Computation $\tau_2'$ is the subsequence of all actions of $\tau_2$ produced by instructions from $\aprog$ and by their clones in the code copy of the attacker, except the store actions to auxiliary address.
These store actions constitute $\tau_4'$.
Finally, $\tau_3'$ is the subsequence of all helper actions of $\tau_3$ produced by clones of instructions from $\aprog$.
We also strip the suffix $\delay$ from the addresses of load and store actions in $\tau_2'$ and $\tau_4'$.

That $\tau$ is a computation of program $\aprog$ follows from the fact that $\tau_\attack$ is executable. 
We just removed actions produced by the instrumentation and replaced buffering by delaying of store actions; we did not change any data dependencies.
The delaying of $\attackstore\cdot\tau_4'$ past $\attackload$ is possible because the attacker did not execute memory fences between $\attacklocalstore$ and $\attackload$, as guaranteed by \eqref{Equation:Fence}.

Let us check that $\tau$ is a TSO witness (Figure~\ref{Figure:ShapeViolation}).
\wita{} holds as in $\tau$ indeed only the attacker delays stores.
The first delayed store $\attackstore$ is an instance of $\storeinst$, load $\attackload$ is an instance of $\loadinst$ and is the last action of the attacker that is overstepped by delayed stores, \witb{} holds.
For each $\anact$ in $\attackload\cdot\tau_3\cdot\attackstore$ it holds $\attackload\happensbefore^*\anact$, \witc{}.
This is by construction of helpers in accordance with Lemma~\ref{Lemma:Access}.
Computation $\tau_4'$ consists only of the stores delayed by the attacker, \witd{}.
\wite{} holds due to the check in \eqref{Equation:LoadInst}.
So, $\tau$ is a TSO witness for attack $\attack$.\\[0.2cm]
\textbf{Completeness.}
Suppose there is a TSO witness $\tau$ for attack $\attack$ as in Figure~\ref{Figure:ShapeViolation}:
$$\tau = \tau_1\cdot\issueact_{\attackstore}\cdot\tau_2\cdot\attackload\cdot\tau_3\cdot\attackstore\cdot\tau_4.$$
We show that the instrumented program has an execution that leads to a goal state.
In the beginning, the instrumented attacker and helper threads execute instructions of the original program, namely those in $\tau_1$.
The helpers actually execute these actions instrumented by \eqref{Equation:HelperOriginal}, i.e., with an additional assert.
These conditionals are executable because the attacker did not yet set variable $\hb$.

Then the attacker executes $\semattackermove{\storeinst}$ ($\storeinst$ is the instruction that produced $\issueact_{\attackstore}$ in $\tau$) and enters the code copy.
Now all its stores will be executed on auxiliary addresses, as defined in \eqref{Equation:StoreInst} and \eqref{Equation:Store}.
This means, they stay invisible to the helper threads as they were in the computation of the original program.
Also, the instrumentation of loads \eqref{Equation:Load} makes sure that they read buffered values, if they exist.
Altogether this preserves the data dependencies from the original computation.

So the attacker executes the instructions that lie in $\tau_2$, instrumented by $\semattacker{-}$. 
Note that $\tau_2$ does not contain memory fences, otherwise $\attackstore$ could not have been delayed past $\attackload$ in $\tau$.
Therefore, \eqref{Equation:Fence} cannot provoke a block of the attacker.
The helpers still execute the actions of the original program, instrumented by \eqref{Equation:HelperOriginal}.
Finally, the attacker executes $\loadinst$ which produced $\attackload$ in $\tau$, Equation~\eqref{Equation:LoadInst}.
This is possible due to \wite{}.

All actions in $\tau_3$ belong to helpers.
By \witc{}, they are in happens-before relation with $\attackload$.
Therefore, due to the instrumentation based on Lemma~\ref{Lemma:Access}, the helpers are able to enter their code copies, \eqref{Equation:HelperLoadMove} and \eqref{Equation:HelperStoreMove}, and execute the instructions that produced $\tau_3$.
Note that the instrumentation of the code copy for helpers does not introduce any conditionals that could block the execution.

At least one of the helper's actions in $\tau_3$ performs a load or a store to the address used in $\attackstore$. 
Otherwise, \witc{} would not hold ($\attackload$ and the delayed write of the attacker use different addresses by \wite{}). 
When performing the action in the instrumented program, the helper will set the $\hb$-variable for the address used in $\attackstore$ to a non-zero value, Equations~\eqref{Equation:HelperStore} and \eqref{Equation:HelperLoad}.
Therefore, at the next step the helper will be able to set $\successvar$ to $\mytrue$ in accordance with \eqref{Equation:HelperSuccessCheck} and make the instrumented program reach a goal state.\qed
\end{proof}

\section{Decidability and Complexity}
The reductions of robustness to reachability and parameterized reachability are independent of the number of addresses and the structure of the data domain.
Hence, without further assumptions the resulting reachability queries cannot be guaranteed to be decidable.
We now discuss conditions on address space and data domain that render
robustness decidable.
Note that we only have to restrict these two dimensions.
The instrumentation copes with the unbounded size store buffers.
Moreover, we choose the verification technology so that it handles the unbounded number of threads required in parameterized reachability.
\paragraph{Parallel Programs with Finite Domains}
Consider a parallel program over a finite data domain, and hence finite address space.
In this setting robustness is \pspace-complete \cite{BMM11}.
Our earlier proof is of complexity-theoretic nature: based on enumeration and not meant to be implemented.
The instrumentation in this paper yields an alternative proof of membership in \pspace\ that is conceptually simpler and allows us to reuse all techniques that have been developed for finite state verification.
\begin{theorem}
Robustness for parallel programs over finite domains is \pspace-complete.
\end{theorem}
\paragraph{Parameterized Programs with Finite Domains}
Consider parameterized programs over finite domains.
In this setting, decidability of robustness was open (our techniques from \cite{BMM11} do not carry over).
With Theorem~\ref{Theorem:Parameter}, we can now solve the problem and establish decidability.
The key observation is that threads in instance programs never use their identifiers, simply because they are copies of the same source code.
This means there is no need to track the identity of threads, it is sufficient to count how many instances of a thread are in each state --- a technique known as counter abstraction \cite{GermanSistla1992}.
Using this technique, we can reformulate the reachability problem for parameterized programs as a coverability problem for Petri nets. We briefly recall the basics on Petri nets.

\paragraph{Definitions}
A \emph{Petri net} is a triple $\petrinet=(\places, \transitions, \weightfun)$ where $\places$ is a finite set of \emph{places}, $\transitions$ is a finite set of \emph{transitions} with $S\cap T = \emptyset$, and $\weightfun\colon (\places\times\transitions)\cup(\transitions\times\places)\to\Nat$ is a \emph{weight function}. A \emph{marking} is a function that assigns a natural number to each place: $\amarking\colon\places\to\Nat$. A \emph{marked Petri net} is a pair $(\petrinet,\amarking_0)$ of a Petri net and an \emph{initial marking} $\amarking_0$.
A transition $\atransition\in\transitions$ is \emph{enabled in marking $\amarking$} if $\amarking(\aplace)\geq\weightfun(\aplace,\atransition)$ for all $\aplace\in\places$.
The \emph{firing relation} $\fires{}\subseteq\Nat^{\power{\places}}\times\transitions\times\Nat^{\power{\places}}$ contains a tuple $(\amarking_1,\atransition,\amarking_2)$ if transition $\atransition$ is enabled in $\amarking_1$ and for all $\aplace\in\places$ we have $\amarking_2(\aplace) = \amarking_1(\aplace) - \weightfun(\aplace, \atransition) + \weightfun(\atransition,\aplace)$.
We also write $\amarking_1\fires{\atransition}\amarking_2$.
We extend the firing relation to sequences of transitions.

We say that a marking $\amarking$ is \emph{reachable} in a marked Petri net $(\petrinet,\amarking_0)$ if there is a transition sequence $\sigma\in\transitions^*$, such that $\amarking_0\fires{\sigma}\amarking$. A marking $\amarking$ is \emph{coverable} if there is a reachable marking $\amarking'$ so that $\amarking(\aplace)'\geq\amarking(\aplace)$ for all $\aplace\in\places$.
\begin{lemma}[\cite{Rackoff1978}]
\label{Lemma:CoverabilityForPetriNetsIsDecidable}
The problem to determine whether a marking $\amarking$ is coverable in a marked Petri net $(\petrinet,\amarking_0)$ is decidable.
\end{lemma}

\paragraph{Reduction of parameterized reachability to Petri net coverability}
Let $\aprog$ be a parameterized program with finite data domain $\datadomain$.
We define a Petri net $\petrinet=(\places, \transitions, \weightfun)$ simulating the program.

For each pair of address and value $(\anaddr,\aval)\in\datadomain\times\datadomain$ we create a place $\aplace_{\anaddr,\aval}$. 
These places represent the state of the global memory: $\amarking(\aplace_{\anaddr,\aval}) = 1$ corresponds to $\valconf(\anaddr) = \aval$.

For each thread $\athread_i$ that declares registers $\listof{\areg_i}$ and has labels $\listof{\alab_i}$ we create places $\aplace_{\alab,\listof{\aval}}$ for all $\alab\in\listof{\alab_i}$ and all $\listof{\aval}\in\datadomain^{\power{\listof{\areg_i}}}$. 
These places encode the number of thread instances in the given control state that have the given register valuation.

For each thread $\athread_i$ we create a transition $\atransition_i$.
Let $\alab_{0,i}$ be the initial label of $\athread_i$, $\listof{\areg_i}$ be the registers declared by $\athread_i$, and $\listof{\aval_{0,i}}$ be a zero vector of length $\power{\listof{\areg_i}}$.
Then we set $\weightfun(\atransition_i, \aplace_{\alab_{0,i},\listof{v_{0,i}}}) = 1$.
Transition $\atransition_i$ effectively spawns an arbitrary number of copies of thread $\athread_i$ that are all in the initial state.

Next we create transitions that simulate the instructions in each thread.
We explain the construction for load instructions. 
The other instructions are handled along similar lines. 
Consider thread $\athread_i$ with registers $\listof{\areg_i}$ and a labelled load instruction $\mathit{linst}=\thetransition{\alab_1}{\alab_2}{\theload{\areg}{\afun_\anaddr(\listof{\areg_\anaddr})}}$\ .
For each value $\aval\in\datadomain$ and for each vector $\listof{\aval_{\mathit{reg}}}\in\datadomain^{\power{\listof{\areg_i}}}$ we create a transition $\atransition=\atransition_{\mathit{linst},\aval,\listof{\aval_\mathit{reg}}}$.
We set $\weightfun(\aplace_{\alab_1,\listof{\aval_{\mathit{reg}}}},\atransition) = \weightfun(\atransition, \aplace_{\alab_2,\listof{\aval_{\mathit{reg}}}'}) = 1$ where $\listof{\aval_{\mathit{reg}}}' = \listof{\aval_{\mathit{reg}}}[\areg := \aval]$.
Let $\anaddr = \afun_\anaddr(\projection{\listof{\aval_{\mathit{reg}}}}{\listof{\areg_\anaddr}})$.
Then we set $\weightfun(\aplace_{\anaddr,\aval},\atransition) = \weightfun(\atransition,\aplace_{\anaddr,\aval}) = 1$.
Transition $\atransition$ is enabled if there is an instance of the thread in control state is $\alab_1$ so that its register valuation is $\listof{\aval_\mathit{reg}}$ and address $\anaddr$ being read holds value $\aval$. Firing the transition only updates the state of the thread instance: its program counter is set to label $\alab_2$, and the value of register $\areg$ is set to $\aval$.

We define the initial marking by $\amarking_0(\aplace_{\anaddr,0})=1$ for all $\anaddr\in\datadomain$, and $\amarking_0(\aplace)=0$ for all other places $\aplace\in\places$. Reaching a goal configuration $\valconf(\successvar)=\mytrue$ in the parameterized program now corresponds to covering the following marking $\amarking_{\successvar}$ in the resulting Petri net: $\amarking_{\successvar}(\aplace_{\successvar,\mytrue}) = 1$ and $\amarking_{\successvar}(\aplace) = 0$ for all other places $\aplace\in\places$. 
Combining this reduction with Lemma~\ref{Lemma:CoverabilityForPetriNetsIsDecidable} gives Theorem~\ref{Theorem:ParameterizedRobustnessIsDecidable}.
\begin{theorem}\label{Theorem:ParameterizedRobustnessIsDecidable}
Robustness for parameterized programs over finite domains is decidable.
\end{theorem}

\paragraph{Lower bound}
The upper bound on robustness for parameterized programs depends on the data domain.
Interestingly, an \expspace\ lower bound already holds for domains with two values.
The proof reduces the coverability problem in Petri nets to robustness of parameterized programs.
\expspace-hardness of coverability is a classic result by Lipton \cite{Lipton1976}.
That we can restrict ourselves to domains with two elements means the control flow in a parameterized program is expressive enough to encode the Petri net behaviour.

The idea behind the construction is to take thread instances as tokens.
Each thread has a label for each place in the Petri net, plus an additional label that indicates the token is currently not in use.
The Petri net transitions are mimicked by a controller thread.
It serialises the reading and writing of tokens, checks the coverability query, and then enters a non-robust situation.
To read and write tokens, the controller communicates with the token threads via the  memory.
The construction requires locked instructions, which are immediate to add to our programming model.
\begin{theorem}
Robustness for parameterized programs is \expspace-hard, for any domain with at least two elements.
\end{theorem}

\end{document}